\documentclass[11pt, twoside]{article}
\usepackage{fancyhdr}
\usepackage[colorlinks,citecolor=blue,urlcolor=blue,linkcolor=blue,bookmarks=false]{hyperref}
\usepackage{amsfonts,epsfig,graphicx}
\usepackage{afterpage}
\usepackage{amsmath,amssymb, amsthm, bm, dsfont}
\usepackage{fullpage}
\usepackage[normalem]{ulem}
\usepackage{enumerate}
\usepackage[T1]{fontenc}
\usepackage{algorithm}
\usepackage{algorithmic}
\usepackage{subfigure}
\usepackage[authoryear, round, colon]{natbib}
\usepackage{multirow}

\def\eps{\epsilon}

\long\def\comment#1{}

\newcommand{\bel}{\begin{eqnarray}\label}
\newcommand{\eel}{\end{eqnarray}}
\newcommand{\bes}{\begin{eqnarray*}}
	\newcommand{\ees}{\end{eqnarray*}}

\newcommand{\abs}[1]{\left|#1\right|}

\def\##1\#{\begin{align}#1\end{align}}
\def\$#1\${\begin{align*}#1\end{align*}}

\newcommand{\e}{\bm{e}}

\newcommand{\g}{\bm{g}}

\newcommand{\x}{\bm{x}}
\newcommand{\y}{\bm{y}}
\newcommand{\z}{\bm{z}}

\newcommand{\Ab}{\mathbf{A}}

\newcommand{\Db}{\mathbf{D}}

\newcommand{\Fb}{\mathbf{F}}

\newcommand{\Hb}{\mathbf{H}}
\newcommand{\Ib}{\mathbf{I}}

\newcommand{\Mb}{\mathbf{M}}
\newcommand{\Nb}{\mathbf{N}}

\newcommand{\Pb}{\mathbf{P}}

\newcommand{\Tb}{\mathbf{T}}
\newcommand{\Ub}{\mathbf{U}}

\newcommand{\cC}{\mathcal{C}}

\newcommand{\cE}{\mathcal{E}}

\newcommand{\cG}{\mathcal{G}}
\newcommand{\cH}{\mathcal{H}}
\newcommand{\cI}{\mathcal{I}}

\newcommand{\cL}{\mathcal{L}}

\newcommand{\cN}{\mathcal{N}}

\newcommand{\EE}{\mathbb{E}}

\newcommand{\RR}{\mathbb{R}}
\newcommand{\SSS}{\mathbb{S}}

\newcommand{\bzeta}{\bm{\zeta}}

\newcommand{\bxi}{\bm{\xi}}

\newcommand{\bPhi}{\bm{\Phi}}
\newcommand{\bPsi}{\bm{\Psi}}

\newcommand{\Phib}{\mathbf{\Phi}}
\newcommand{\Psib}{\mathbf{\Psi}}

\newcommand{\argmin}{\mathop{\mathrm{argmin}}}

\newcommand{\sign}{\mathop{\mathrm{sign}}}
\newcommand{\rank}{\mathop{\mathrm{rank}}}
\newcommand{\supp}{\mathop{\mathrm{supp}}}
\newcommand{\median}{\mathop{\mathrm{median}}}

\global\long\def\K{\bm{K}}

\global\long\def\A{\bm{A}}

\global\long\def\x{\bm{x}}
\global\long\def\g{\bm{g}}
\global\long\def\c{\bm{c}}
\global\long\def\y{\bm{y}}
\global\long\def\z{\bm{z}}
\global\long\def\e{\bm{e}}

\global\long\def\d{\bm{d}}
\global\long\def\b{\bm{b}}
\global\long\def\a{\bm{a}}

\global\long\def\u{\bm{u}}
\global\long\def\v{\bm{v}}

\global\long\def\M{\bm{M}}

\newtheoremstyle{mytheoremstyle} 
    {\topsep}                    
    {\topsep}                    
    {\normalfont}                   
    {}                           
    {\bfseries}                   
    {.}                          
    {.5em}                       
    {}  

\theoremstyle{mytheoremstyle}

\ifx\BlackBox\undefined
\newcommand{\BlackBox}{\rule{1.5ex}{1.5ex}}  
\fi

\ifx\QED\undefined
\def\QED{~\rule[-1pt]{5pt}{5pt}\par\medskip}
\fi

\ifx\proof\undefined
\newenvironment{proof}{\par\noindent{\bf Proof\ }}{\hfill\BlackBox\\[2mm]}
\fi

\ifx\theorem\undefined
\newtheorem{theorem}{Theorem}
\fi
\ifx\example\undefined
\newtheorem{example}[theorem]{Example}
\fi
\ifx\property\undefined
\newtheorem{property}{Property}
\fi
\ifx\lemma\undefined
\newtheorem{lemma}[theorem]{Lemma}
\fi
\ifx\proposition\undefined
\newtheorem{proposition}[theorem]{Proposition}
\fi
\ifx\remark\undefined
\newtheorem{remark}[theorem]{Remark}
\fi
\ifx\corollary\undefined
\newtheorem{corollary}[theorem]{Corollary}
\fi
\ifx\definition\undefined
\newtheorem{definition}[theorem]{Definition}
\fi
\ifx\conjecture\undefined
\newtheorem{conjecture}[theorem]{Conjecture}
\fi
\ifx\fact\undefined
\newtheorem{fact}[theorem]{Fact}
\fi
\ifx\claim\undefined
\newtheorem{claim}[theorem]{Claim}
\fi
\ifx\assumption\undefined
\newtheorem{assumption}[theorem]{Assumption}
\fi
\numberwithin{equation}{section}
\numberwithin{theorem}{section}

\begin{document}

\title{\huge Binary Embedding: Fundamental Limits and Fast Algorithm}

\author{
Xinyang Yi\\
{The University of Texas at Austin}\\
{yixy@utexas.edu}
\and
Constantine Caramanis\\
{The University of Texas at Austin}\\
{constantine@utexas.edu} 
\and
Eric Price\\
{The University of Texas at Austin}\\
{ecprice@cs.utexas.edu} }
\date{}

\maketitle

\begin{abstract}
  Binary embedding is a nonlinear dimension reduction methodology
  where high dimensional data are embedded into the Hamming cube while
  preserving the structure of the original space. Specifically, for an
  arbitrary $N$ distinct points in $\SSS^{p-1}$, our goal is to encode
  each point using $m$-dimensional binary strings such that we can
  reconstruct their geodesic distance up to $\delta$ uniform
  distortion. Existing binary embedding algorithms either lack
  theoretical guarantees or suffer from running time
  $O\big(mp\big)$. We make three contributions: (1) we establish a
  lower bound that shows any binary embedding oblivious to the set of
  points requires $m = \Omega(\frac{1}{\delta^2}\log{N})$ bits and a
  similar lower bound for non-oblivious embeddings into Hamming
  distance;
  \sout{(2) we propose a novel fast binary embedding algorithm with provably
  optimal bit complexity $m =
  O\big(\frac{1}{\delta^2}\log{N}\big)$ and near linear running time
  $O(p \log p)$ whenever $\log N \ll \delta \sqrt{p}$, with a slightly
  worse running time for larger $\log N$;} (3) we also provide an
  analytic result about embedding a general set of points $K \subseteq
  \SSS^{p-1}$ with even infinite size.  Our theoretical findings are
  supported through experiments on both synthetic and real data sets.

  \textbf{[Note: a previous version of this paper also included a
    claimed fast upper bound for certain parameter regimes.  The proof
    of this had an error, as pointed out in~\cite{dirksen2018fast};
    the same paper also presents a correct algorithm for the
    setting.]}
\end{abstract}

\section{Introduction}

Low distortion embeddings that transform high-dimensional points to low-dimensional space have played an important role in dealing with storage, information retrieval and machine learning problems for modern datasets. Perhaps one of the most famous results along these lines is the Johnson-Lindenstrauss (JL) lemma~\citet{johnson1984extensions}, which shows that $N$ points can be embedded into a $O\big(\delta^{-2}\log{N}\big)$-dimensional space while preserving pairwise Euclidean distance up to $\delta$-Lipschitz distortion. This $\delta^{-2}$ dependence has been shown to be information-theoretically optimal \citet{alon2003problems}. Significant work has focused on fast algorithms for computing the embeddings, e.g., \citep{ailon2006approximate, krahmer2011new, ailon2013almost, cheraghchi2013restricted, nelson2014new}. 

More recently, there has been a growing interest in designing binary codes for high dimensional points with low distortion, i.e., embeddings into the binary cube \citep{weiss2009spectral, raginsky2009locality,salakhutdinov2009semantic, liu2011hashing, gong2011iterative, yu2014circulant}. Compared to JL embedding, embedding into the binary cube (also called binary embedding) has two advantages in practice: (i)  As each data point is represented by a binary code, the disk size for storing the entire dataset is reduced considerably. (ii) Distance in binary cube is some function of the Hamming distance, which can be computed quickly using computationally efficient bit-wise operators. As a consequence, binary embedding can be applied to a large number of domains such as biology, finance and computer vision where the data are usually high dimensional. 

While most JL embeddings are linear maps, any binary embedding is fundamentally a nonlinear transformation. As we detail below, this nonlinearity poses significant new technical challenges for both upper and lower bounds. In particular, our understanding of the landscape is significantly less complete. To the best of our knowledge, lower bounds are not known; embedding algorithms for infinite sets have distortion-dependence $\delta$ significantly exceeding their finite-set counterparts; and perhaps most significantly, there are no fast (near linear-time) embedding algorithms with strong performance guarantees. As we explain below, this paper contributes to each of these three areas. First, we detail some recent work and state of the art results.

\noindent {\bf Recent Work}. A common approach pursued by several existing works, considers the natural extension of JL embedding techniques via one bit quantization of the projections:
\begin{equation}
	\b(\x) = \sign(\Ab \x), \label{eq:quant}
\end{equation}  
where $\x \in \RR^{p}$ is input data point, $\Ab \in \RR^{m \times p}$ is a projection matrix and $\b(\x)$ is the embedded binary code. In particular, \citet{jacques2011robust} shows when each entry of $\Ab$ is generated independently from $\cN(0,1)$, with $m > \frac{1}{\delta^2}\log{N}$ it with high probability achieves at most $\delta$  (additive) distortion for $N$ points. Work in \citet{plan2014dimension} extend these results to arbitrary sets $K \subseteq \SSS^{p-1}$ where $|K|$ can be infinite. They prove that the embedding with $\delta$-distortion can be obtained when $m \gtrsim w(K)^2/\delta^6$ where $w(K)$ is the {\em Gaussian Mean Width} of $K$. It is unknown whether the unusual $\delta^{-6}$ dependence is optimal or not. Despite provable sample complexity guarantees, one bit quantization of random projection as in \eqref{eq:quant}, suffers from $O\big(mp\big)$ running time for a single point. This quadratic dependence can result in a prohibitive computational cost for high-dimensional data. Analogously to the developments in ``fast'' JL embeddings, there are several algorithms proposed to overcome this computational issue. Work in \citet{gong2013learning} proposes a bilinear projection method. By setting $m = O(p)$, their method reduces the running time from $O(p^2)$ to $O(p^{1.5})$. More recently, work in \citet{yu2014circulant} introduces a circulant random projection algorithm that requires running time $O\big(p\log p\big)$. While these algorithms have reduced running time, as of yet they come without performance guarantees: to the best of our knowledge, the measurement complexities of the two algorithms are still unknown. Another line of work considers learning binary codes from data by solving certain optimization problems \citep{weiss2009spectral, salakhutdinov2009semantic, norouzi2012Hamming, yu2014circulant}. Unfortunately, there is no known provable bits complexity result for these algorithms. It is also worth noting that \citet{raginsky2009locality} provide a binary code design for preserving shift-invariant kernels. Their method suffers from the same quadratic computational issue compared with the fully random Gaussian projection method. 

Another related dimension reduction technique is  locality sensitive hashing (LSH) where the goal is to compute a discrete data structure such that similar points are mapped into the same bucket with high probability (see, e.g., \citet{andoni2006near}). The key difference is that LSH preserves short distances, but binary embedding preserves both short and far distances.  For points that are far apart, LSH only cares that the hashings are different while binary embedding cares how different they are.

\noindent {\bf Contributions of this paper.} In this paper, we address several unanswered problems about binary embedding. We provide lower bounds for both data-oblivious and data-aware embeddings; we provide a fast algorithm for binary embedding; and finally we consider the setting of infinite sets, and prove that in some of the most common cases we can improve the state-of-the-art sample complexity guarantees by a factor of $\delta^{-2}$:
\begin{enumerate}
\item We provide two lower bounds for binary embeddings.  The first
  shows that any method for embedding and for recovering a distance
  estimate from the embedded points that is independent of the data
  being embedded must use $\Omega(\frac{1}{\delta^2}\log{N})$ bits.
  This is based on a bound on the communication complexity of Hamming
  distance used by~\cite{jayram2013optimal} for a lower bound on the
  ``distributional'' JL embedding.  Separately, we give a lower bound
  for arbitrarily data-dependent methods that embed into (any function
  of) the Hamming distance, showing such algorithms require $m =
  \Omega(\frac{1}{\delta^2\log{(1/\delta)}}\log{N})$.  This bound is
  similar to~\citet*{alon2003problems} which gets the same result for
  JL, but the binary embedding requires a different construction.

\item \sout{We provide the first provable fast algorithm with optimal measurement complexity  $O\big(\frac{1}{\delta^2}\log{N}\big)$.}
  \textbf{[A previous version of this paper included an incorrect claimed result here.]}
  \iffalse The proposed algorithm has running time  $O\big(\frac{1}{\delta^2}\log{\frac{1}{\delta}}\log^2{N}\log{p}\log^3 \log N + p \log p\big)$ thus has almost linear time complexity when $\log N \lesssim \delta \sqrt{p}$. Our algorithm is based on two key novel ideas. First, our similarity is based on the median Hamming distance of sub-blocks of the binary code; second, our new embedding takes advantage of a {\em pair-wise independence argument} of Gaussian Toeplitz projection that could be of independent interest. }
\fi

\item For arbitrary set $K \subseteq \SSS^{p-1}$ and the fully random Gaussian projection algorithm, we prove that $m  = O (w(K^{+})^2/\delta^4)$ is sufficient to achieve $\delta$ uniform distortion. Here $K^{+}$ is an {\em expanded} set of $K$. Although in general $K \subseteq K^{+}$ and hence $w(K) \leq w(K^{+})$, for interesting $K$ such as sparse or low rank sets, one can show $w(K^{+}) = \Theta(w(K)) \ll p$. Therefore applying our theory to these sets results in an improved dependence on $\delta$ compared to a recent result in \citet{plan2014dimension}. See Section \ref{delta-uniform embedding} for a detailed discussion.
\end{enumerate}

\noindent {\bf Notation.} We use $[n]$ to denote natural number set $\{1,2,\ldots,n\}$. For natural numbers $a < b$, let $[a,b]$ denote the consecutive set $\{a,a+1,\ldots,b\}$. A vector in $\RR^{n}$ is denoted as $\x$ or equivalently $(x_1,x_2,\ldots,x_n)^{\top}$. We use $\x_{\cI}$ to denote the sub-vector of $\x$ with index set $\cI \subseteq [n]$. We denote entry-wise vector multiplication as $\x\odot\y = (x_1y_1,x_2y_2,\ldots,x_ny_n)^{\top}$. A matrix is typically denoted as $\Mb$. Term $(i,j)$ of $\Mb$ is denoted as $\Mb_{i,j}$. Row $i$ of $\Mb$ is denoted as $\Mb_i$.  An $n$-by-$n$ identity matrix is denoted as $\Ib_{n}$. For two random variables $X,Y$, we denote the statement that $X$ and $Y$ are independent as $X \bot Y$. For two binary strings $\a,\b \in \{0,1\}^{m}$, we use $d_{\cH}(\a,\b)$ to denote the normalized Hamming distance, i.e., $d_{\cH}(\a,\b) := \frac{1}{m}\sum_{i=1}^{m}\mathds{1}(a_i \ne b_i)$.

\section{Organization, Problem Setup and Preliminaries}

In this section, we state our problem formally, give some key definitions and present a simple (known) algorithm that sets the stage for the main results of this paper. The algorithm (Algorithm \ref{alg:urp}), discussed in detail below, is simply the one-bit quantization of a standard JL embedding. Its performance {\em on finite} sets is easy to analyze, and we state it in Proposition \ref{thm:urp} below. Three important questions remain unanswered: (i) Lower Bounds -- is the performance guaranteed by Proposition \ref{thm:urp} optimal? We answer this affirmatively in Section \ref{sec:lower bound}. (ii) Fast Embedding -- whereas Algorithm \ref{alg:urp} is quadratic (depending on the product $mp$), fast JL algorithms are nearly linear in $p$; does something similar exist for binary embedding? We develop a new algorithm in Section \ref{sec:fast algorithm} that addresses the complexity issue, while at the same time guaranteeing $\delta$-embedding with dimension scaling that matches our lower bound. Interestingly, a key aspect of our contribution is that we use a slightly modified similarity function, using the median of the normalized Hamming distance on sub-blocks. (iii) Infinite Sets -- recent work analyzing the setting of infinite sets $K \subseteq \mathbb{S}^{p-1}$ shows a dependence of $\delta^{-6}$ on the distortion. Is this optimal? We show in Section \ref{delta-uniform embedding} that in many settings this can be improved by a factor of $\delta^{-2}$. In Section \ref{sec:numerical results}, we provide numerical results. We give most proofs in Section \ref{sec:proofs}.

\subsection{Problem Setup}
Given a set of $p$-dimensional points, our goal is to find a transformation $f: \mathbb{R}^p \mapsto \{0,1\}^m$ such that the Hamming distance (or other related, easily computable metric) between two binary codes is close to their similarity in the original space. We consider points on the unit sphere $\mathbb{S}^{p-1}$ and use the normalized geodesic distance (occasionally, and somewhat misleadingly, called cosine similarity) as the input space similarity metric. For two points  $\x, \y \in \mathbb{R}^{p}$, we use $d(\x,\y)$ to denote the geodesic distance, defined as
\[
	d(\x,\y) := \frac{\angle(\x/\|\x\|_2,\y/\|\y\|_2)}{\pi},
\]
where $\angle(\cdot,\cdot)$ denotes the angle between two vectors. For $\x,\y \in \SSS^{p-1}$, the metric $d(\x,\y)$ is proportional to the length of the shortest path connecting $\x,\y$ on the sphere.
 
Given the success of JL embedding, a natural approach is to consider the one bit quantization of a random projection: 
\begin{equation} \label{framework}
	\b = \text{sign}(\Ab \x),
\end{equation}
where $\Ab$ is some random projection matrix. Given two points $\x, \y$ with embedding vectors $\b$, and $\c$, we have $b_i \ne c_i$ if and only if $\big\langle \Ab_i , \x \big\rangle \big\langle \Ab_i, \y\big\rangle < 0$.
The traditional metric in the embedded space has been the so-called normalized Hamming distance, which we done by $d_{\Ab}(\x,\y)$ and is defined as follows.
\begin{equation} \label{AHamming}
	d_{\Ab}(\x,\y) := \frac{1}{m}\sum_{i=1}^{m} \mathds{1}\bigg\{\text{sign}\big(\big\langle \Ab_i,\x \big\rangle\big) \ne \text{sign}\big(\big\langle \Ab_i,\y \big\rangle\big)\bigg\}.
\end{equation}

\begin{definition} ($\delta$-{uniform} Embedding) 
	Given a set $K \subseteq \mathbb{S}^{p-1}$ and projection matrix $\Ab \in \RR^{m\times p}$, we say the embedding $\b = \text{sign}(\Ab \x)$  provides a $\delta$-uniform embedding for points in $K$ if 
	\begin{equation} \label{uniform embedding}
		\big|d_{\Ab}(\x,\y) - d(\x,\y) \big| \leq \delta, \; \forall \;x, y \in K.
	\end{equation}
\end{definition}

Note that unlike for JL, we aim to control {\em additive} error instead of {\em relative} error. Due to the inherently limited resolution of binary embedding, controlling relative error would force the embedding dimension $m$ to scale inversely with the minimum distance of the original points, and in particular would be impossible for any infinite set.

\subsection{Uniform Random Projection} \label{sec:uniform random projection}
\begin{algorithm}[!htb]
	\caption{Uniform Random Projection}
	\label{alg:urp}
	\begin{algorithmic}[1]
		\INPUT  Finite number of points $K = \{\x_i\}_{i=1}^{|K|}$ where $K \subseteq \SSS^{p-1}$, embedding target dimension $m$.
		\STATE Construct  matrix $\Ab \in \RR^{m \times p}$ where each entry $\Ab_{i,j}$ is drawn independently from $\cN(0,1)$.
		\FOR{$i = 1, 2, \ldots, |K|$}
		\STATE $\b_i \leftarrow \sign(\Ab \x_i)$.
		\ENDFOR
		\OUTPUT $\{\b_i\}_{i=1}^{|K|}$
		\end{algorithmic}
\end{algorithm}

Algorithm \ref{alg:urp} presents \eqref{framework} formally, when $\Ab$ is an i.i.d.\ Gaussian random matrix, i.e., $\Ab_i \sim \cN(0,\Ib_p)$ for any $i \in [m]$. It is easy to observe that for two fixed points $\x,\y \in \SSS^{p-1}$ we have 
\begin{equation}
	\EE\bigg( \mathds{1}\bigg\{\sign\big(\big\langle \Ab_i, \x\big\rangle\big) \ne \sign\big(\big\langle \Ab_i, \y\big\rangle\big)\bigg\} \bigg) = d(\x,\y),\; \forall \;i \in [m].
\end{equation}
The above equality has a geometric explanation: each $\Ab_i$ actually represents a uniformly distributed random hyperplane in $\RR^{p}$. Then $\sign\big(\big\langle \Ab_i, \x\big\rangle\big) \ne \sign\big(\big\langle \Ab_i, \y\big\rangle\big)$ holds if and only if hyperplane $\Ab_i$ intersects the arc between $\x$ and $\y$. In fact, $d_{\Ab}(\x,\y)$ is equal to the fraction of such hyperplanes. Under such uniform tessellation, the probability with which the aforementioned event occurs is $d(\x,\y)$.  Applying Hoeffding's inequality and probabilistic union bound over $N^2$ pairs of points, we have the following straightforward guarantee. 

\begin{proposition} \label{thm:urp}
	Given a set $K \subseteq \SSS^{p-1}$ with finite size $|K|$, consider Algorithm \ref{alg:urp} with $m \geq c(1/\delta^2)\log{|K|}$. Then with probability at least $1 - 2\exp(-\delta^2m)$, we have 
	\[
		\big| d_{\Ab}(\x,\y) - d(\x,\y) \big| \leq \delta, \; \forall \; \x,\y \in K.
	\]
Here $c$ is some absolute constant.
\end{proposition}
\begin{proof}
	The proof idea is standard and follows from the above; we omit the details.
\end{proof}

\section{Main Results}

We now present our main results on lower bounds, on fast binary embedding, and finally, on a general result for infinite sets.

\subsection{Lower Bounds}\label{sec:lower bound}

We offer two different lower bounds.  The first shows that any
embedding technique that is oblivious to the input points must use
$\Omega(\frac{1}{\delta^2} \log N)$ bits, regardless of what method is
used to estimate geodesic distance from the embeddings.  This shows
that uniform random projection and our fast binary embedding achieve
optimal bit complexity (up to constants).  The bound follows from
results by~\cite{jayram2013optimal} on the communication complexity of
Hamming distance.

\begin{theorem} \label{thm:oblivious lower bound}
	Consider any distribution on embedding functions $f: \mathbb{S}^{p-1} \to \{0,
	1\}^m$ and reconstruction algorithms $g: \{0, 1\}^m \times \{0, 1\}^m
	\to \RR$ such that for any $\x_1, \dotsc, \x_N \in \SSS^{p-1}$ we have
	\[
		\big|g(f(\x_i), f(\x_j)) - d(\x_i, \x_j)\big| \leq \delta
	\]
	for all $i, j \in [N]$ with probability $1 - \eps$.  Then $m =
	\Omega(\frac{1}{\delta^2}\log (N/\eps))$.
\end{theorem}
\begin{proof}
	See Section \ref{proof:thm:oblivious lower bound} for detailed proof.
\end{proof}

One could imagine, however, that an embedding could use knowledge of
the input point set to embed any specific set of points into a
lower-dimensional space than is possible with an oblivious algorithm.
In the Johnson-Lindenstrauss setting, \citet*{alon2003problems} showed
that this is not possible beyond (possibly) a $\log (1/\delta)$
factor.  We show the analogous result for binary embeddings.  Relative
to Theorem~\ref{thm:oblivious lower bound}, our second lower bound
works for data-dependent embedding functions but loses a $\log
(1/\delta)$ and requires the reconstruction function to depend only on
the Hamming distance between the two strings.  This restriction is
natural because an unrestricted data-dependent reconstruction function
could simply encode the answers and avoid any dependence on $\delta$.

With the scheme given in (\ref{framework}), choosing $\Ab$ as a fully
random Gaussian matrix yields $d_{\Ab}(\x,\y) \approx
d(\x,\y)$. However, an arbitrary binary embedding algorithm may not
yield a linear functional relationship between Hamming distance and
geodesic distance. Thus for this lower bound, we allow the design of an
algorithm with arbitrary link function $\cL$.

\begin{definition} \label{def:gbep} (Data-dependent binary embedding problem) \\
	Let $\cL: [0,1] \rightarrow [0,1]$ be a monotonic and continuous function. Given a set of points $\x_1,\x_2,...,\x_N \in \SSS^{p-1}$, we say a binary embedding mapping $f$ solves the binary embedding problem in terms of link function $\cL$, if 
	\begin{equation} \label{general_equation}
		\big|d_{\cH}(f(\x_i),f(\x_j)) - \cL\big(d(\x_i,\x_j)\big)\big| \leq \delta, \; \forall \;i ,j \in [N].
	\end{equation}
\end{definition}

Although the choice of $\cL$ is flexible, note that for the same point, we always have $d_{\cH}(f(\x_i),f(\x_i)) = d(\x_i,\x_i) = 0$, thus (\ref{general_equation}) implies $\cL(0) < \delta$. We can just let $\cL(0) = 0$. In particular, we let $\cL_{\rm max} = \cL(1)$. We have the following lower bound:

\begin{theorem} \label{thm:lower bound}
	There exist $2N$ points $\x_1, \x_2,...,\x_{2N} \in \mathbb{S}^{N-1}$ such that for any binary embedding algorithm $f$ on $\{\x_i\}_{i=1}^{2N}$, if it solves the data-dependent binary embedding problem defined in \ref{def:gbep} in terms of link function $\cL$ and any $ \delta \in (0, \frac{1}{16\sqrt{e}}\cL_{\rm max})$, it must satisfy
	\begin{equation} \label{lower bound}
		m \geq  \frac{1}{128e} \bigg(\frac{\cL_{\rm max}}{\delta}\bigg)^2\frac{\log{N} }{\log{\frac{\cL_{\rm max}}{2\delta}}}.
	\end{equation}
\end{theorem}
\begin{proof}
	See Section \ref{proof:thm:lower bound} for detailed proof.
\end{proof}

\begin{remark} 
	We make two remarks for the above result. (1) When $\cL_{\rm max}$ is some constant, our result implies that for general $N$ points, any binary embedding algorithm (even data-dependent ) must have $\Omega(\frac{1}{\delta^2\log{\frac{1}{\delta}}}\log N)$ number of measurements. This is analogous to Alon's lower bound in the JL setting. It is worth highlighting two differences: (i) The JL setting considers the same metric (Euclidean distance) for both the input and the embedded spaces. In binary embedding, however, we are interested in showing the relationship between Hamming distance and geodesic distance. (ii) Our lower bound is applicable to a broader class of binary embedding algorithms as it involves arbitrary, even data-dependent, link function $\cL$. Such an extension is not considered in the lower bound of JL. (2) The stated lower bound only depends on $\cL_{\rm max}$ and does not depend on any curvature information of $\cL$. The constraint $\cL_{\rm max} > 16\sqrt{e}\delta$ is critical for our lower bound to hold, but some such restriction is necessary because for $\cL_{\rm max} < \delta$, we are able to embed all points into just one bit. In this case $d_{\cH}(f(\x_i),f(\x_j)) = 0$ for all pairs and condition (\ref{general_equation}) would hold trivially.
\end{remark}

\subsection{Fast Binary Embedding} \label{sec:fast algorithm}

\textbf{[Note: this section contains an incorrect proof of a result.
  We leave the section here because some of the intermediate lemmas
  are used by \cite{dirksen2018fast}, which pointed out the error in
  this proof and provided a correct algorithm.  Incorrect
  Theorem/Lemma statements are noted as such.]}

In this section, we present a novel fast binary embedding algorithm. We then establish its theoretical guarantees. There are two key ideas that we leverage: (i) instead of normalized Hamming distance, we use a related metric, the median of the normalized Hamming distance applied to sub-blocks; and (ii) we show a key pair-wise independence lemma for partial Gaussian Toeplitz projection, that allows us to use a concentration bound that then implies nearness in the median-metric we use.

\subsubsection{Method}

Our algorithm builds on sub-sampled Walsh-Hadamard matrix and partial Gaussian Toeplitz  matrices with random column flips. In particular, an $m$-by-$p$ partial Walsh-Hadamard matrix has the form
\begin{equation} \label{fastJL}
	\Phib :=  \Pb\cdot\Hb\cdot\Db.
\end{equation}
The above construction has three components. We characterize each term as follows:
\begin{itemize}
	\item Term $\Db$ is a $p$-by-$p$ diagonal matrix with diagonal terms $\{\zeta_i\}_{i=1}^{p}$ that are drawn from i.i.d.\ Rademacher sequence, i.e, for any $i \in [p]$, $\Pr(\zeta_i = 1) = \Pr(\zeta_i = -1) = 1/2$.
	\item Term $\Hb$ is a $p$-by-$p$ scaled Walsh-Hadamard matrix such that $\Hb^{\top}\Hb = \Ib_{p}$.
	\item Term $\Pb$ is an $m$-by-$p$ sparse matrix where one entry of each row is set to be $1$ while the rest are $0$. The nonzero coordinate of each row is drawn independently from uniform distribution. In fact, the role of $\Pb$ is to randomly select $p$ rows of $\Hb\cdot\Db$. 
\end{itemize}
An $m$-by-$n$ partial Gaussian Toeplitz matrix has the form 
\begin{equation} \label{partialGaussianRandom}
	\Psib :=  \Pb\cdot\Tb\cdot\Db.
\end{equation}
We introduce each term as follows:
\begin{itemize}
	\item Term $\Db$ a is $n$-by-$n$ diagonal matrix with diagonal terms $\{\zeta_i\}_{i=1}^{n}$ that are drawn from i.i.d.\ Rademacher sequence.
	\item Term $\Tb$ is a $n$-by-$n$ Toeplitz matrix constructed from $(2n-1)$-dimensional vector $\g$ such that $\Tb_{i,j} = g_{i-j+n}$ for any $i,j \in [n]$. In particular, $\g$ is drawn from $\cN(0,\Ib_{2n-1})$.
	\item Term $\Pb$ is an $m$-by-$n$ sparse matrix where $\Pb_{i} = \e_i^{\top}$ for any $i \in [m]$. Equivalently, we use $\Pb$ to select the first $m$ rows of $\Tb \Db$. It's worth to note we actually only need to select any distinct $m$ rows. 
\end{itemize}

With the above constructions in hand, we present our fast algorithm in Algorithm \ref{alg:fbe}. At a high level, Algorithm \ref{alg:fbe} consists of two parts: First, we apply column flipped partial Hadamard transform to convert $p$-dimensional point into $n$-dimensional intermediate point. Second, we use $B$ independent $(m/B)$-by-$n$ partial Gaussian Toeplitz matrices and sign operator to map an intermediate point into $B$ blocks of binary codes.  In terms of similarity computation for the embedded codes, we use the median of each block's normalized Hamming distance. In detail, for $\b,\c \in \{0,1\}^{m}$, $B$-wise normalized Hamming distance is defined as 
\begin{equation} \label{B-wise Hamming}
	d_{\cH}(\b,\c; B) := \median \bigg(  \bigg\{ d_{\cH}\big(\b_{T_i}, \c_{T_i}\big) \bigg\}_{i=0}^{B-1}\bigg)
\end{equation}
where $T_i = [i+1,i+m/B]$.

It is worth noting that our first step is one construction of fast JL
transform.  In fact any fast JL transform would work for our
construction, but we choose a standard one with real value: based on
\citet{rudelson2008sparse,cheraghchi2013restricted, krahmer2011new},
it is known that with $m = O\big(\epsilon^{-2}\log{N}\log p\log^3(\log
N)\big)$ measurements, a subsampled Hadamard matrix with column flips
becomes an $\epsilon$-JL matrix for $N$ points.

The second part of our algorithm follows framework (\ref{framework}). By choosing a Gaussian random vector in each row of $\Psi$, from our previous discussion in Section \ref{sec:uniform random projection}, the probability that such a hyperplane intersects the arc between two points is equal to their geodesic distance. Compared to a fully random Gaussian matrix, as used in Algorithm \ref{alg:urp}, the key difference is that the hyperplanes represented by rows of $\Psi$ are not independent to each other; this imposes the main analytical challenge.

\begin{algorithm}[!htb]
	\caption{Fast Binary Embedding}
	\label{alg:fbe}
	\begin{algorithmic}[1]
		\INPUT  Finite number of points $\{\x_i\}_{i=1}^{N}$ where each point $\x_i \in \SSS^{p-1}$, embedded dimension $m$, intermediate dimension $n$, number of blocks $B$.
		\STATE Draw a $n$-by-$p$ sub-sampled Walsh-Hadamard matrix $\Phib$ according to \eqref{fastJL}. Draw $B$ independent partial Gaussian Toeplitz matrices $\big\{\Psib^{(j)}\big\}_{j=1}^{B}$ with size $(m/B)$-by-$n$ according to \eqref{partialGaussianRandom}.
		\STATE \COMMENT{{\em Part I: Fast JL}}
		\FOR{$i = 1, 2, \ldots, N$}
		\STATE $\y_i \leftarrow \Phib \cdot \x_i$.
		\ENDFOR
		\STATE \COMMENT{{\em Part II: Partial Gaussian Toeplitz Projection}}
		\FOR{$i = 1, 2, \ldots, N$}
		\FOR{$j = 1, 2, \ldots, B$}
		\STATE $\c_j \leftarrow \sign\big(\Psib^{(j)} \cdot \y_i\big)$.
		\ENDFOR
		\STATE $\b_{i} \leftarrow [\c_{1};\c_{2};\ldots;\c_{B} ]$
		\ENDFOR
		\OUTPUT $\{\b_i\}_{i=1}^{N}$
	\end{algorithmic}
\end{algorithm}
\subsubsection{Analysis}
We give the analysis for Algorithm \ref{alg:fbe}. We first review a well known result about fast JL transform.
\begin{lemma} \label{lem:fastJL}
	Consider the column flipped partial Hadamard matrix defined in (\ref{fastJL}) with size $m$-by-$p$. For $N$ points $\x_1,\x_2,...,\x_N \in \SSS^{p-1}$, let $\y_i = \sqrt{\frac{p}{m}}\Phi(\bzeta)\cdot\x_i, \; \forall \; i \in [N]$. For some absolute constant $c$, suppose $m \geq c\delta^{-2}\log{N}\log{p}\log^3(\log N)$, then with probability at least $0.99$, we have that for any $i,j \in [N]$
	\begin{equation} \label{c1}
		\big| \|\y_i - \y_j\|_2 - \|\x_i - \x_j\|_2 \big| \leq  \delta\|\x_i - \x_j\|_2,
	\end{equation}
	and for any $i \in [N]$
	\begin{equation} \label{c2}
		\big| \|\y_i\|_2  - 1 \big| \leq \delta.
	\end{equation}
\end{lemma}
\begin{proof}
	It can be proved by combining Theorem 14 in \citet{cheraghchi2013restricted} and Theorem 3.1 in \citet{krahmer2011new}.
\end{proof}
The above result suggests that the first part of our algorithm reduces the dimension while preserving well the Euclidean distance of each pair. Under this condition, all the pairwise geodesic distances are also well preserved as confirmed by the following result.
\begin{lemma} \label{lem:followFastJL}
	Consider the set of embedded points $\{\y_i\}_{i=1}^N$ defined in Lemma \ref{lem:fastJL}. Suppose conditions (\ref{c1})-(\ref{c2}) hold with $\delta > 0$. Then for any $i,j \in [N]$,
	\begin{equation} \label{bound1}
		\big|d(\y_i,\y_j) - d(\x_i,\x_j)\big| \leq C\delta
	\end{equation}
	holds with some absolute constant $C$.
\end{lemma}
\begin{proof} 
	We postpone the proof to Appendix \ref{proof:lem:followFastJL}.
\end{proof}
The next result is our independence lemma, and is one of the key technical ideas that make our result possible. The result shows that for any fixed $\x$, Gaussian Toeplitz projection (with column flips) plus $\sign(\cdot)$ generate pair-wise independent binary codes.

\begin{lemma} \label{lem:pairwise_independent}
	Let $\g \sim \cN(0,\Ib_{2n-1})$, $\bzeta = \{\zeta_i\}_{i=1}^{i=n}$ be an i.i.d.\ Rademacher sequence. Let $\Tb$ be a random Toeplitz matrix constructed from $\g$ such that $\Tb_{i,j} = g_{i-j+n}$. Consider any two distinct rows of $\Tb$ say $\bxi$, $\bxi'$. For any two fixed vectors $\x,\y \in \mathbb{R}^{n}$, we define the following random variables 
\begin{align*}
	X = \sign\big\langle \bxi \odot\bzeta,\; \x \big\rangle, & \;\;X' = \sign\big\langle \bxi' \odot \; \bzeta, \x \big\rangle; \\ 
	Y = \sign\big\langle \bxi \odot\bzeta,\; \y \big\rangle, & \;\;\;Y' = \sign\big\langle \bxi' \odot \; \bzeta, \y \big\rangle.
\end{align*}
We have
\[
	X \bot X', X \bot Y', Y \bot X', Y \bot Y'.
\]
\end{lemma}
\begin{proof} 
	See Section \ref{proof:lem:pairwise_indepedent} for detailed proof.
\end{proof}

We are ready to prove the following result about Algorithm \ref{alg:fbe}.
\begin{theorem} \label{thm:fbe} \textbf{(Incorrect)}
Consider Algorithm \ref{alg:fbe} with random matrices $\bPhi,\;\bPsi$ defined in (\ref{fastJL}) and (\ref{partialGaussianRandom}) respectively. For finite number of points $ \{\x_i\}_{i=1}^{N}$, let $\b_i$ be the  binary codes of $\x_i$ generated by Algorithm \ref{alg:fbe}. Suppose we set
\[
	B \geq c\log{N},\;n \geq c'(1/\delta^2)\log{N}\log{p}\log^3(\log N),\; n \geq m/B \geq c''(1/\delta^2),
\]
with some absolute constants $c, c', c''$, then with probability at least $0.98$, we have that for any $i,j \in [N]$
\[
	\big|d_{\cH}(\b_i,\b_j; B) - d(\x_i,\x_j)\big| \leq \delta.
\]
Similarity metric $d_{\cH}(\cdot,\cdot;B)$ is the median of normalized Hamming distance defined in (\ref{B-wise Hamming}).
\end{theorem}
\begin{proof} 
	See Section \ref{proof:thm:fbe} for detailed proof.
\end{proof}

The above result suggests that the measurement complexity of our fast algorithm is $O\big(\frac{1}{\delta^2}\log{N}\big)$ which matches the performance of Algorithm \ref{alg:urp} based on fully random matrix. Note that this measurement complexity can not be improved significantly by any data-oblivious binary embedding with any similarity metric, as suggested by Theorem \ref{thm:oblivious lower bound}.

\noindent{\bf Running time:} The first part of our algorithm takes time $O\big(p\log p\big)$. Generating a single block of binary codes from partial Toeplitz matrix takes time $O\big(n\log(\frac{1}{\delta})\big)$\footnote{Matrix-vector multiplication for $m$-by-$n$ partial Toeplitz matrix can be implemented in running time  $O\big(n\log m\big)$.}. Thus the total running time is $O\big(B n \log{\frac{1}{\delta}} + p\log p\big) = O\big(\frac{1}{\delta^2}\log{\frac{1}{\delta}}\log^2 N \log p \log^3(\log N) + p\log p\big)$. By ignoring the polynomial $\log\log$ factor, the second term $O\big(p\log p\big)$ dominates when $\log N  \lesssim \delta\sqrt{p/\log{\frac{1}{\delta}}} $.

\noindent{\bf Comparison to an alternative algorithm:} Instead of utilizing the partial Gaussian Toeplitz projection, an alternative method, to the best of our knowledge not previously stated,  is to use fully random Gaussian projection in the second part of our algorithm. We present the details in Algorithm \ref{alg:alter-fbe}. By combining Proposition \ref{thm:urp} and Lemma \ref{lem:fastJL}, it is straightforward to show this algorithm still achieves the same measurement complexity $O\big(\frac{1}{\delta^2}\log N\big)$. The corresponding running time is $O\big(\frac{1}{\delta^4}\log^2 N \log p\log^3(\log N) + p \log p\big)$, so it is fast when $\log N \lesssim \delta^2\sqrt{p}$. Therefore our algorithm has an improved dependence on $\delta$. This improvement comes from fast multiplication of partial Toeplitz matrix and a pair-wise independence argument shown in Lemma \ref{lem:pairwise_independent}.

\begin{algorithm}[!htb]
	\caption{Alternative Fast Binary Embedding}
	\label{alg:alter-fbe}
	\begin{algorithmic}[1]
		\INPUT  Finite number of points $\{\x_i\}_{i=1}^{N}$ where each point $\x_i \in \SSS^{p-1}$, embedded dimension $m$, intermediate dimension $n$.
		\STATE Draw a $n$-by-$p$ sub-sampled Walsh-Hadamard matrix $\Phib$ according to \eqref{fastJL}. Construct $m$-by-$n$ matrix $\Ab$ where each entry is drawn independently from $\cN(0,1)$.
		\FOR{$i = 1, 2, \ldots, N$}
		\STATE $\b_{i} \leftarrow \sign(\Ab\Phib\x_i)$
		\ENDFOR
		\OUTPUT $\{\b_i\}_{i=1}^{N}$
	\end{algorithmic}
\end{algorithm}

\subsection{$\delta$-uniform Embedding for General $K$} \label{delta-uniform embedding}

In this section, we turn back to the fully random projection binary embedding (Algorithm \ref{alg:urp}). Recall that in Proposition \ref{thm:urp}, we show for finite size $K$,  $m = O(\frac{1}{\delta^2}\log{|K|})$ measurements are sufficient to achieve $\delta$-uniform embedding. For general $K$, the challenge is that there might be an infinite number of distinct points in $K$, so Proposition \ref{thm:urp} cannot be applied. In proving the JL lemma for an infinite set $K$, the standard technique is either constructing an $\epsilon$-net of $K$ or reducing the distortion to the deviation bound of a Gaussian process. However, due to the non-linearity essential for binary embedding, these techniques cannot be directly extended to our setting. Therefore strengthening Proposition \ref{thm:urp} to infinite size $K$ imposes significant technical challenges. Before stating our result, we first give some definitions. 
\begin{definition} (Gaussian mean width)
Let $\g \sim \cN(0,\Ib_{p})$. For any set $K \subseteq \SSS^{p-1}$, the Gaussian mean width of $K$ is defined as 
\[
	w(K) := \mathbb{E}_{\g} \sup_{\x \in K} \big| \big\langle \g, \x\big\rangle\big|.
\]
\end{definition}
Here, $w(K)^2$ measures the effective dimension of set $K$. In the trivial case $K = \SSS^{p-1}$, we have $w(K)^2 \lesssim p$. However, when $K$ has some special structure, we may have $w(K)^2 \ll p$. For instance, when $K = \{\x \in \SSS^{p-1}: |\supp(\x)| \leq s \}$, it has been shown that $w(K) = \Theta(\sqrt{s\log(p/s)})$ (see Lemma 2.3 in \citet{plan2013robust}).

For a given $\delta$, we define $K_{\delta}^+$, the {\em expanded version} of $K \subseteq \SSS^{p-1}$ as:
\begin{equation}
\label{def:K+}
	K^{+}_{\delta} := K\bigcup \big\{\z \in \SSS^{p-1}: \z = \frac{\x - \y}{\|\x - \y\|_2}, \; \forall \; \x,\y \in K\; \text{if} \;\delta^2 \leq \|\x - \y\|_2 \leq \delta\big\}.
\end{equation}
In other words, $K^+_{\delta}$ is constructed from $K$ by adding the normalized differences between pairs of points in $K$ that are within $\delta$ but not closer than $\delta^2$. Now we state the main result as follows.
\begin{theorem} \label{thm:general_K}
	Consider any $K \subseteq \SSS^{p-1}$. Let $\Ab \in \RR^{m \times p}$ be an i.i.d.\ Gaussian matrix where each row $\Ab_i \sim \cN(0,\Ib_{p})$. For any two points $\x,\y \in K$, $\d_{\Ab}(\x,\y)$ is defined in (\ref{AHamming}). Expanded set $K^+_{\delta}$ is defined in (\ref{def:K+}). When
	\[
		m \geq c\frac{w(K^+_{\delta})^2}{\delta^4},
	\]
	with some absolute constant $c$, then we have that
	\[
		\sup_{\x,\y \in K} \big| d_{\Ab}(\x,\y) - d(\x,\y)\big| \leq \delta
	\]
	holds with probability at least $1 - c_1\exp(-c_2\delta^2m)$ where $c_1,c_2$ are absolute constants.
\end{theorem}
\begin{proof}
	See Section \ref{proof:thm:general_K} for detailed proof.
\end{proof}
\begin{remark}
We compare the above result to Theorem 1.5 from the recent paper \cite{plan2014dimension} where it is proved that for $m \gtrsim w(K)^2/\delta^6$, Algorithm \ref{alg:urp} is guaranteed to achieve $\delta$-uniform embedding for general $K$. Based on definition (\ref{def:K+}), we have 
\[
	w(K) \leq w(K^+_{\delta}) \leq \frac{1}{\delta^2}w(K-K) \lesssim \frac{1}{\delta^2}w(K).
\] 
Thus in the worst case, Theorem \ref{thm:general_K} recovers the previous result up to a factor $\frac{1}{\delta^2}$. More importantly, for many interesting sets one can show $w(K^+_{\delta}) \lesssim w(K)$; in such cases, our result leads to an improved dependence on $\delta$. We give several such examples as follows:
\begin{itemize}
	\item {\bf Low rank set.} For some $\Ub \in \RR^{p\times r}$ such that $\Ub^{\top}\Ub = \Ib_{r}$, let $K = \{\x \in \SSS^{p-1}: \x = \Ub\c, \; \forall\; \c \in \SSS^{r-1} \}$. We simply have $K = K^{+}_{\delta}$ and $w(K) \lesssim \sqrt{r}$. Our result implies $m = O\big(r/\delta^4\big)$.
	\item {\bf Sparse set.} $K = \{\x \in \SSS^{p-1}: |\supp(\x)| \leq s\}$. In this case we have $K^+_{\delta} \subseteq \{\x \in \SSS^{p-1}: |\supp(\x)| \leq 2s\}$. Therefore $w(K^{+}_{\delta}) = \Theta(\sqrt{s\log(p/s)})$. Our result implies $m = O\big(\frac{s\log(p/s)}{\delta^4}\big)$.
	\item {\bf Set with finite size.} $|K| < \infty$. As $w(K) \lesssim \sqrt{\log|K|}$ and $|K^+_{\delta}| \leq 2|K|$,  our result implies $m = O\big(\log{|K|}/\delta^4\big)$. We thus recover Proposition \ref{thm:urp} up to factor $1/\delta^2$. 
\end{itemize}
\end{remark}
Applying the result from \citet{plan2014dimension} to the above sets implies similar results but the dependence on $\delta$ becomes $1/\delta^6$.

\section{Numerical Results} \label{sec:numerical results}

In this section, we present the results of experiments  we conduct to validate our theory and compare the performance of the following three algorithms we discussed: uniform random projection (URP) (Algorithm \ref{alg:urp}), fast binary embedding (FBE) (Algorithm \ref{alg:fbe}) and the alternative fast binary embedding (FBE-2) (Algorithm \ref{alg:alter-fbe}). We first apply these algorithms to synthetic datasets. In detail, given parameters $(N, p)$, a synthetic dataset is constructed by sampling $N$ points from $\SSS^{p-1}$ uniformly at random. Recall that $\delta$ is the maximum embedding distortion among all pairs of points. We use $m$ to denote the number of binary measurements. Algorithm FBE needs parameters $n, B$, which are intermediate dimension and number of blocks respectively. Based on Theorem \ref{thm:fbe}, $n$ is required to be  proportional to $m$ (up to some logarithmic factors) and $B$ is required to be proportional to $\log N$. We thus set $n \approx 1.3m$, $B \approx 1.8\log N$. We also set $n \approx 1.3m$ for FBE-2. In addition, we fix $p = 512$. We report our first result showing the functional relationship between $(m,N,\delta)$ in Figure \ref{fig:syn_sim}. In particular, panel \ref{fig:synthetic1} shows the the change of distortion $\delta$ over the number of measurements $m$ for fixed $N$. We observe that, for all the three algorithms,  $\delta$ decays with $m$ at the rate predicted by Proposition \ref{thm:urp} and Theorem \ref{thm:fbe}. Panel \ref{fig:synthetic2} shows the empirical relationship between $m$ and $\log N$ for fixed $\delta$. As predicted by our theory (lower bound and upper bound), $m$ has a linear dependence on $\log N$.  

\begin{figure}[ht]
	\centering
	\subfigure[$N = 300$]{
		\centering
		\label{fig:synthetic1}
		\includegraphics[width=0.45\columnwidth]{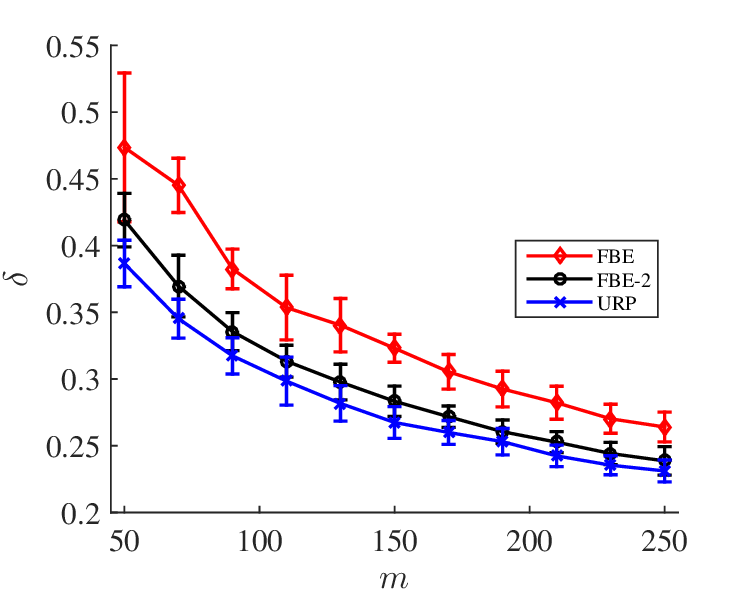}
	}
	\hfill
	\subfigure[$\delta = 0.3$]{
		\centering
		\label{fig:synthetic2}
		\includegraphics[width=0.45\columnwidth]{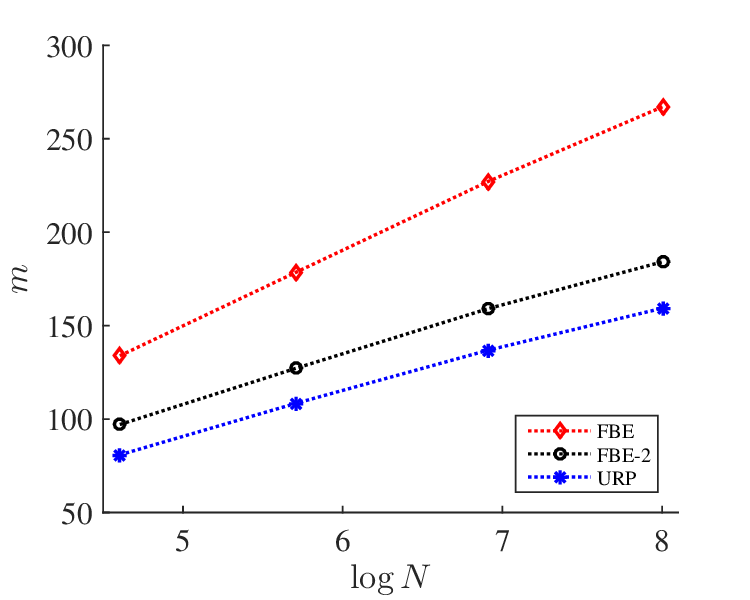}
	}
	\caption{Results on synthetic datasets. {\bf (a)} Each point, along with the standard deviation represented by the error bar, is an average of $50$ trials each of which is based on a fresh synthetic dataset with size $N = 300$ and newly constructed embedding mapping. {\bf (b)} Each point is computed by slicing at $\delta = 0.3$ in similar plots like (a) but with the corresponding $N$.}
	\label{fig:syn_sim}
\end{figure}

\begin{figure}[ht]
	\centering
	\subfigure[$m = 5000$]{
		\centering
		\includegraphics[width=0.31\columnwidth]{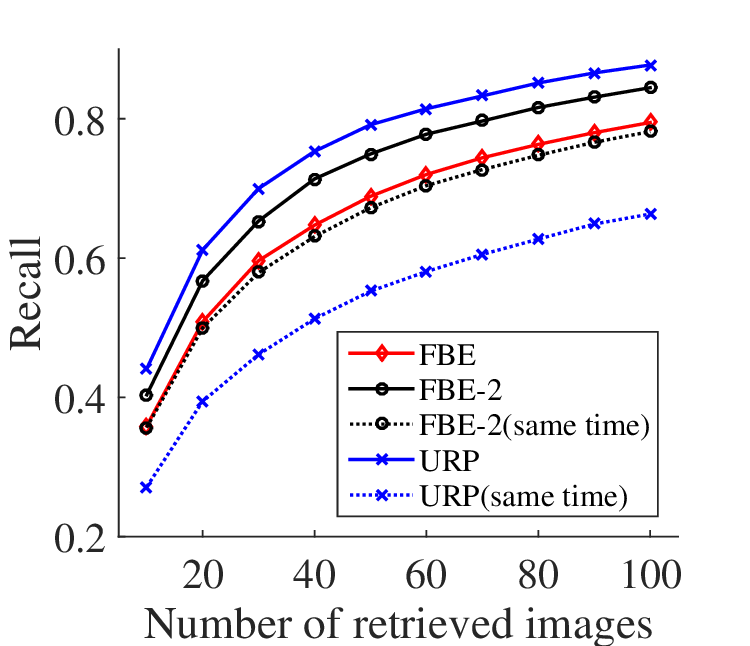}
	}
	\hfill
	\subfigure[$m = 10000$]{
		\centering
		\includegraphics[width=0.31\columnwidth]{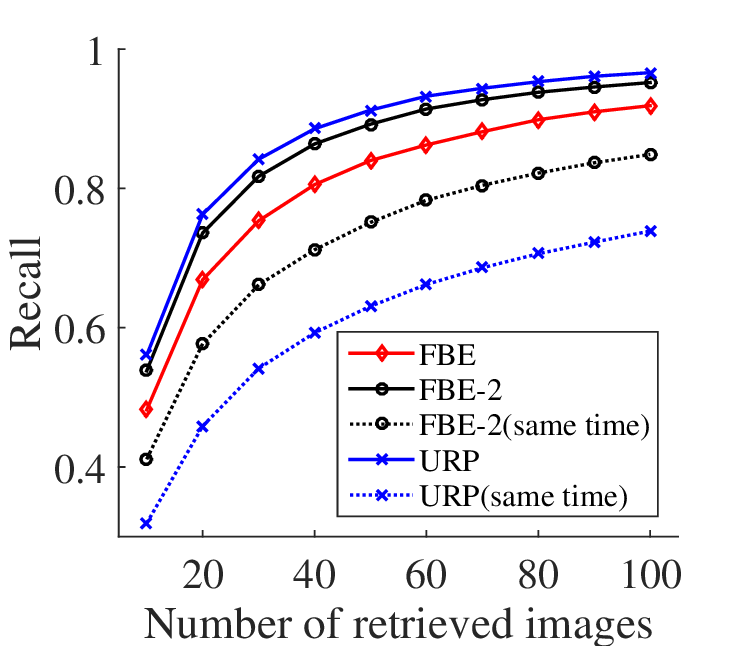}
	}
	\hfill
	\subfigure[$m = 15000$]{
		\centering
		\includegraphics[width=0.31\columnwidth]{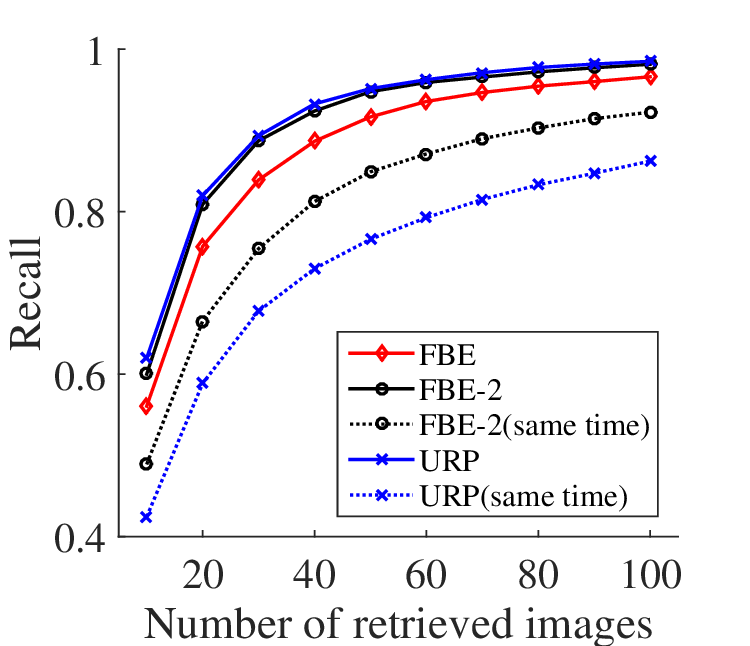}
	}
	\caption{Image retrieval results on Flickr-25600. Each panel presents the recall for specified number of measurements $m$. Black and blue dot lines are respectively the recall of FBE-2 and URP with less number of measurements but the same running time as FBE.}
	\label{fig:real_sim}
\end{figure}

A popular application of binary embedding is image retrieval, as considered in \citep{gong2011iterative, gong2013learning, yu2014circulant}. We thus conduct an experiment on the Flickr-25600 dataset that consists of $10k$ images from Internet. Each image is represented by a $25600$-dimensional normalized Fisher vector. We take $500$ randomly sampled images as query points and leave the rest as base for retrieval. The {\em relevant images} of each query are defined as its $10$ nearest neighbors based on geodesic distance. Given $m$, we apply FBE, FBE-2 and URP to convert all images into $m$-dimensional binary codes. In particular, we set $B = 10$ for FBE and $n \approx 1.3m$ for FBE and FBE-2. Then we leverage the corresponding similarity metrics, \eqref{B-wise Hamming} for FBE and Hamming distance for FBE-2 and URP, to retrieve the nearest images for each query. The performance of each algorithm is characterized by {\em recall}, i.e., the number of retrieved {\em relevant} images divided by the total number of relevant images. We report our second result in Figure \ref{fig:real_sim}. Each panel shows the average recall of all queries for a specified $m$. We note that FBE-2, as a fast algorithm, performs as well as URP with the same number of measurements. In order to show the running time advantage of our fast algorithm FBE, we also present the performance of FBE-2 and URP with fewer measurements such that they can be computed with the same time as FBE. As we observe, with large number of measurements, FBE-2 and URP perform marginally better than FBE while FBE has a significant improvement over the two algorithms under identical time constraint.

\section{Proofs} \label{sec:proofs}
\subsection{Proof of Data-Oblivious Lower Bound (Theorem \ref{thm:oblivious lower bound})}\label{proof:thm:oblivious lower bound}

The proof of the data-oblivious lower bound is based on a lower bound
for one-way communication of Hamming distance due
to~\cite{jayram2013optimal}.

\begin{definition}[One-way communication of Hamming distance]
  In the one-way communication model, Alice is given $\a \in \{0,1\}^n$
  and Bob is given $\b \in \{0, 1\}^n$.  Alice sends Bob a message $\c
  \in \{0, 1\}^m$, and Bob uses $\b$ and $\c$ to output a value $x \in
  \RR$.  Alice and Bob have shared randomness.

  Alice and Bob solve the $(\delta, \eps)$ additive Hamming distance estimation
  problem if $\abs{x - d_\cH(\a, \b)} \leq \delta$ with probability
  $1-\eps$.
\end{definition}

The result proven in~\cite{jayram2013optimal} is a lower bound for the
\emph{multiplicative} Hamming distance estimation problem, but their
techniques readily yield a bound for the additive case as well:
\begin{lemma}\label{l:hamming-hard}
  Any algorithm that solves the $(\delta, \eps)$ additive Hamming
  distance estimation problem must have $m = \Omega((1/\delta^2) \log
  (1/\eps))$ as long as this is less than $n$.

\end{lemma}
\begin{proof}
  We apply Lemma~3.1 of~\cite{jayram2013optimal} with parameters
  $\alpha = 2$, $p = 1$, $b=1$, $\varepsilon = \delta$, and $\delta =
  \eps$.  This encodes inputs from a problem they prove is hard
  (augmented indexing on large domains) to inputs appropriate for
  Hamming estimation.  In particular, for $n' =
  O(\frac{1}{\delta^2}\log(1/\eps))$ it gives a distribution on $(\a,
  \b) \in \{0,1\}^{n'}\times \{0, 1\}^{n'}$ that are divided into ``NO''
  and ``YES'' instances, such that:
  \begin{itemize}
  \item From the reduction, distinguishing NO instances from YES
    instances with probability $1-\eps$ requires Alice to send $m =
    \Omega(\frac{1}{\delta^2}\log(1/\eps))$ bits of communication to Bob.
  \item In NO instances, $d_\cH(\a, \b) \geq \frac{1}{2}(1 - \delta/3)$.
  \item In YES instances, $d_\cH(\a, \b) \leq \frac{1}{2}(1 - 2\delta/3)$.
  \end{itemize}
  First, suppose $n=n'$.  Then since solving the additive Hamming distance
  estimation problem with $\delta/12$ accuracy would distinguish NO
  instances from YES instances, it must involve $m =
  \Omega(\frac{1}{\delta^2}\log(1/\eps))$ bits of communication.

  For $n > n'$, simply duplicate the coordinates of $a$ and $b$
  $\lfloor{n/n'}\rfloor$ times, and zero-pad the remainder.  Less than
  half the coordinates are then part of the zero-padding, so the gap
  between YES and NO instances remains at least $\delta/12$ and a
  protocol for the $(\delta/24, \eps)$ additive Hamming distance
  estimation problem requires $m =
  \Omega(\frac{1}{\delta^2}\log(1/\eps))$ as desired.
\end{proof}

With this in hand, we can prove Theorem~\ref{thm:oblivious lower bound}:
\begin{proof}[Proof of Theorem~\ref{thm:oblivious lower bound}]
  We reduce one-way communication of the $(\delta, \eps)$ additive
  Hamming distance estimation problem to the embedding problem.  Let
  $\a, \b \in \{0, 1\}^p$ be drawn from the hard instance for the
  communication problem defined in Lemma~\ref{l:hamming-hard}.
  Linearly transform them to $\u, \v \in \SSS^{p-1}$ via $\u =
  (2\cdot\a-\bm{1})/\sqrt{p}$, $\v = (2\cdot\b-\bm{1})/\sqrt{p}$.  We have that $\langle \u, \v \rangle =
  1 - 2d_\cH(\a, \b)$, so
  \[
  d(\u, \v) = 1 - \frac{\arccos(\langle \u, \v \rangle)}{\pi} = 1 - \frac{\arccos(1 - 2d_\cH(\a, \b))}{\pi}
  \]
  or
  \[
  d_{\cH}(\a, \b) = \frac{1}{2}(1-\cos(\pi - \pi d(\u, \v)))
  \]
  Given an estimate of $d(\u, \v)$, we can therefore get an estimate of
  $d_{\cH}(\a, \b)$.  In particular, since $\abs{\cos'(x)} \leq 1$, if
  we learn $d(\u, \v)$ to $\pm \delta$ then we learn $d_{\cH}(\a, \b)$ to
  $\pm \delta \frac{\pi}{2}$.

  For now, consider the case of $N = 2$.  Consider an oblivious
  embedding function $f: \mathbb{S}^{p-1} \to \{0, 1\}^m$ and
  reconstruction algorithm $g: \{0, 1\}^m \times \{0, 1\}^m \to \RR$
  that has
  \[
  \abs{g(f(\u), f(\v)) - d(\u, \v)} \leq \delta\frac{2}{\pi}
  \]
  with probability $1-\eps$ on the distribution of inputs $(\u,\v)$.  We
  can solve the one-way communication problem for Hamming distance
  estimation by Alice sending $f(\u)$ to Bob, Bob learning $d(\u, \v)
  \approx g(f(\u), f(\v))$, and then computing $d_\cH(\a, \b)$ to $\pm
  \delta$.  By the lower bound for this problem, any such $f$ and $g$
  must have $m = \Omega(\frac{1}{\delta^2} \log \frac{1}{\eps})$,
  proving the result for $N = 2$ (after rescaling $\delta$).

  For general $N$, we draw instances $(\u_1, \v_1), (\u_2, \v_2), \dotsc,
  (\u_{N/2}, \v_{N/2})$ independently from the hard instance for binary
  embedding of $N = 2$ and $\eps' = 4\eps/N$.  Consider an oblivious
  embedding function $f: \mathbb{S}^{p-1} \to \{0, 1\}^m$ and
  reconstruction algorithm $g: \{0, 1\}^m \times \{0, 1\}^m \to \RR$
  that has for all $i \in [N/2]$ that
  \[
  \abs{g(f(\u_i), f(\v_i)) - d(\u_i, \v_i)} \leq \delta
  \]
  with probability $1-\eps$ on this distribution.  Define $\alpha$ to
  be the probability that $\abs{g(f(\u_i), f(\v_i)) - d(\u_i, \v_i)} \leq
  \delta$ for any particular $i$.  Because $f$ and $g$ are oblivious
  and the different instances are independent, we have the probability
  that all instances succeed is $\alpha^{N/2} \geq 1 - \eps$, so
  \[
  \alpha > (1-\eps)^{2/N} > 1- 4\eps/N.
  \]
  In particular, this means $f$ and $g$ solve the hard instance of
  binary embedding and $N = 2$, $\eps' = 4\eps/N$.  By the above lower
  bound for $N=2$, this means
  \[
  m = \Omega(\frac{1}{\delta^2} \log (N/\eps))
  \]
  as desired.
\end{proof}

\subsection{Proof of Data-Dependent Lower Bound (Theorem \ref{thm:lower bound})}\label{proof:thm:lower bound}

We need a few ingredients to show the lower bound. First, we define a matrix that is close to identity matrix.
\begin{definition} ($(\delta_1, \delta_2)$-near identity matrix)
Symmetric matrix $\Mb \in \mathbb{R}^{p\times p}$ is called a $(\delta_1, \delta_2)$-near identity matrix if it satisfies both of the following conditions:
\[
	1 - \delta_1 \leq  \Mb_{i,i}  \leq 1, \forall \; i \in [p],
\]
\[
	\big|\Mb_{i,j}\big| \leq \delta_2, \forall \; i \ne j \in [p].
\]
\end{definition}

Next we give a lower bound on the rank of $(\delta_1,\delta_2)$-near identity matrix.
\begin{lemma} \label{weak_rank_bound}
	Suppose  positive semidefinite matrix $\Mb \in \mathbb{R}^{p\times p}$ is a $(\delta_1,\delta_2)$-near identity matrix with rank $d$, and $0 < \delta_1,\delta_2 < 1$. Then we have 
	\[
		d \geq \frac{p(1-\delta_1)^2}{1 + (p-1)\delta_2^2}.
	\]
\end{lemma}
\begin{proof} 
	We postpone the proof to Appendix \ref{proof:lem:weak_rank_bound}.
\end{proof}
The above result is weak when it is applied to show our desired lower bound. We still need to make use of the following combinatorial result.

\begin{lemma}  \label{boost_lower_bound}
Suppose matrix $\Mb \in \mathbb{R}^{p\times p}$ has rank $d$. Let $P(x)$ be any degree $k$ polynomial function. Consider matrix $\Nb \in \mathbb{R}^{p\times p}$ defined as $\Nb := P(\Mb)$,
where the $\Nb_{i,j} = P(\Mb_{i,j})$. We have
\[
	\rank(\Nb) \leq {k+d \choose k}.
\]
\end{lemma}
\begin{proof}
	See Lemma 9.2 of \cite{alon2003problems} for a detailed proof.
\end{proof}

Now we are ready to prove Theorem \ref{thm:lower bound}.
\begin{proof}[Proof of Theorem~\ref{thm:lower bound}]
	Let $\e_i$ denote the $i$'th natural basis of $\mathbb{R}^{N}$, i.e., the $i$'th coordinate is $1$ while the rest are all zeros. Consider $N$ points $\{\e_1,\e_2,...,\e_N\}$ and their opposite vectors $\{-\e_1,-\e_2,...,-\e_N\}$. For any binary embedding algorithm $f$, we let
	\[
		\b_i := f(\e_i), \; \forall \; i \in [N],
	\]
	\[
		\c_i := f(-\e_i), \; \forall \; i \in [N].
	\]

Under the condition that $f$ solves the general binary embedding problem with link function $\cL$, we have
\begin{equation} \label{bc}
	\big| d_{\mathcal{H}}(\b_i, \c_i) - \cL\big(d(\e_i,-\e_i)\big) \big| \leq  \delta, \forall \; i \in [N]. 
\end{equation}
As $d(\e_i, -\e_i) = 1$, we have
\begin{equation} \label{tmp22}
	\cL(1) + \delta \geq d_{\mathcal{H}}(\b_i, \c_i) \geq \cL(1) - \delta.
\end{equation}
Similarly, note that 
\begin{equation*}
	d(\e_i,\e_j) = d(\e_i,-\e_j) = d(-\e_i,-\e_j) = \frac{1}{2},\;\forall \;i \ne j,
\end{equation*}
we have $\forall \;i \ne j$
\begin{equation} \label{bb}
	\cL(1/2) - \delta \leq d_{\mathcal{H}}(\b_i, \b_j) \leq \cL(1/2) + \delta, 
\end{equation}
\begin{equation} \label{cc}
	\cL(1/2) - \delta \leq d_{\mathcal{H}}(\c_i, \c_j) \leq \cL(1/2) + \delta,
\end{equation}
\begin{equation} \label{bc2}
	\cL(1/2) - \delta \leq d_{\mathcal{H}}(\b_i, \c_j) \leq \cL(1/2) + \delta.
\end{equation}

From now on, we treat binary strings $\b_i,\c_i$ as vectors in $\RR^{m}$.  Let $\mathbf{B}$ denote the matrix with rows $\b_i$ and $\mathbf{C}$ denote the matrix with rows $\c_i$. Consider the outer product of the difference between $\mathbf{B}$ and $\mathbf{C}$, namely
\[
	\mathbf{M} = (\mathbf{B} - \mathbf{C})(\mathbf{B} - \mathbf{C})^{\top}.
\]
Note that  $\forall \; i \in [N]$,
\[
	\Mb_{i,i}  = \| \b_i - \c_i \|_2^2  =  4m\cdot d_{\mathcal{H}} (\b_i, \c_i)  \geq 4m\big(\cL(1)-\delta\big).
\]
The last inequality follows from (\ref{tmp22}).
For $\forall\; i \ne j$, we have
\begin{align*}
	\Mb_{i,j} & = \big\langle \b_i - \c_i, \b_j - \c_j \big\rangle  = \big\langle \b_i, \b_j\big\rangle + \big\langle \c_i, \c_j\big\rangle - \big\langle \b_i, \c_j\big\rangle - \big\langle \b_j, \c_i\big\rangle \\
	& = 2m\bigg(  d_{\mathcal{H}} (\b_i,\c_j) + d_{\mathcal{H}} (\b_j,\c_i) - d_{\mathcal{H}} (\b_i,\b_j) - d_{\mathcal{H}} (\c_i,\c_j) \bigg),
\end{align*}
where the third equality follows from 
\[
	d_{\mathcal{H}}(\b, \c) = \frac{1}{4m} \big( \|\b\|_2^2 + \|\c\|_2^2 - 2\big\langle \b, \c\big\rangle \big) \; \forall \; \b, \c \in \{-1,1\}^m 
\]
By using (\ref{bb}) to (\ref{bc2}), we have
\[
	\big|\M_{i,j}\big| \leq 8\delta m.
\] 
Therefore, $\frac{1}{4m \cdot (\cL(1)+\delta)}\Mb$ is actually a $\bigg(\frac{2\delta}{\cL(1)},\frac{2\delta}{\cL(1)}\bigg)$-near identity matrix. Consider degree $k$ polynomial $P(z) = z^k$. Let \[
\Nb = P\big(\frac{1}{4m \cdot \cL(1)}\Mb \big).
\]

It is easy to observe that $\Nb$ is a $(\gamma_1,\gamma_2)$-near identity matrix where
\[
	\gamma_1 = 1 - (1 - \frac{2\delta}{\cL(1)})^k,
\]
and
\[
	\gamma_2 = \big(\frac{2\delta}{\cL(1)}\big)^k.
\]
Under the condition $\frac{\delta}{\cL(1)} \leq \frac{1}{4}$, we have
\[
	\gamma_1 = 1 - (1 - \frac{\delta}{\cL(1)})^k \leq 1 - (\frac{1}{2})^{k}.
\]
By setting $k = \frac{1}{2}\frac{\log{N}}{\log{\frac{\cL(1)}{2\delta}}} $, we have
\[
	\gamma_2 \leq \sqrt{\frac{1}{N}}.
\]
We apply Lemma \ref{weak_rank_bound} by setting $\delta_1,\delta_2, p$ in the statement to be $\gamma_1,\gamma_2, N$ respectively. We get
\begin{equation} \label{lb}
	\rank(\Nb) \geq \frac{N(\frac{1}{4})^{k}}{1 + (N-1)/N} \geq \frac{1}{2}(\frac{1}{4})^kN \geq (\frac{1}{8})^kN.
\end{equation}
On the other hand, $\frac{1}{4m \cdot \cL(1)}\Mb$ has rank at most $m$. By applying Lemma \ref{boost_lower_bound} we get 
\[
	\rank(\mathbf{\Nb}) \leq {m+k \choose k} \leq \big(\frac{e(m+k)}{k}\big)^k.
\]
Applying the above result and (\ref{lb}) directly yields that
\[
	(N)^{1/k} \leq 8e\frac{m+k}{k}.
\]
When  $k =  \frac{1}{2}\frac{\log{N}}{\log{\frac{\cL(1)}{2\delta}}} $ as we set, $N^{1/k} \geq (\frac{\cL(1)}{2\delta})^2$. Therefore we have
\[
	m \geq \frac{1}{32e}\big(\frac{\cL(1)}{\delta}\big)^2k - k \geq \frac{1}{64e}\big(\frac{\cL(1)}{2\delta}\big)^2k = \frac{1}{128e}\big(\frac{\cL(1)}{\delta}\big)^2\frac{\log{N} }{\log{\frac{\cL(1)}{2\delta}}},
\]
where the second inequality holds when $\big(\frac{\cL(1)}{2\delta}\big)^2 \geq 64e$.
\end{proof}

\subsection{Proofs about Fast Binary Embedding Algorithm}
\subsubsection{Proof of Lemma \ref{lem:pairwise_independent}} \label{proof:lem:pairwise_indepedent}
\begin{proof}
It suffices to prove $X \bot Y'$. One can check similarly that the proof holds for the remaining three results. Note that $X, Y'$ are binary random variables with values $\{-1,1\}$. It is easy to observe both of them are balanced, namely $\Pr(X = 1) = \Pr( Y' = 1) = 1/2$. If $X \bot Y'$, then we have $\Pr(X = Y') = 1/2$. In the reverse direction, suppose $\Pr(X = Y') = 1/2$. First we have
\begin{equation}
	\label{eq1}
	\Pr(X = 1) = \Pr(X = 1, Y' = 1) + \Pr(X = 1, Y' = -1) = 1/2,
\end{equation}
\begin{equation}
	\label{eq2}
	\Pr(Y' = 1) = \Pr(X = 1, Y' = 1) + \Pr(X = -1, Y' = 1) = 1/2.
\end{equation}
Combining the above two results, we have $\Pr(X = 1, Y' = -1) = \Pr(X = -1, Y' = 1)$. Using $ \Pr(X = 1, Y' = -1) + \Pr(X = - 1, Y' = 1) =\Pr(X \ne Y') = 1 - \Pr(X = Y') = \frac{1}{2}$, we thus have $\Pr(X = 1, Y' = -1) = \Pr(X = - 1, Y' = 1) = 1/4$. Plugging the above result into (\ref{eq1}) and (\ref{eq2}) we have $\Pr(X = 1, Y' = 1) = \Pr(X = -1, Y' = -1) = 1/4$.
Thus we have shown 
\[
	\Pr(X = v \big| Y' = u) = \frac{\Pr(X = v, Y' = u)}{\Pr(Y' = u)} = \Pr(X = v), \; \forall \; u,v \in \{-1,1\},
\]
which leads to $X \bot Y'$.

Using the above arguments, we show that $X \bot Y'$ if and only if
\[
	\Pr(X = Y') = 1/2.
\]
Recalling the definition of $X,Y'$, the above condition holds if and only if 
\[
	\Pr\bigg\{ \underbrace{\big\langle \bxi\odot\bzeta, \x\big\rangle \cdot \big\langle \bxi'\odot\bzeta, \y \big\rangle}_{Z} \; \geq \;0\bigg\} = \frac{1}{2}.
\]
Next we prove $Z$ has symmetric distribution around $0$. Let $\cI = [1,n], \cI' = [1,n-\Delta], \cI_0 = [2n-\Delta,2n-1]$ for some natural number $\Delta < n$. Without loss of generality, we assume $\bxi = \g_{\cI}$ and $\bxi' = [\g_{\cI_0};\g_{\cI'}]$.  We split $\cI$ into $T = \lceil \frac{n}{\Delta}\rceil$ consecutive disjoint subsets $\cI_1, \cI_2, \ldots, \cI_{T}$ each of which has size $\Delta$ except $| \cI_{T}| = n - (T-1)\Delta \leq \Delta$. Also, let $\cI_{T-1}'$ contain the first $n - (T-1)\Delta$ entries of $\cI_{T-1}$. Then we have
\begin{equation} \label{Z_formula}
	Z = \bigg( \sum_{i=1}^{T} \big\langle \g_{\cI_i}\odot \bzeta_{\cI_i}, \x_{\cI_i}\big\rangle\bigg) \cdot \bigg( \sum_{i=1}^{T-2}\big\langle \g_{\cI_i}\odot \bzeta_{\cI_{i+1}}, \y_{\cI_{i+1}}\big\rangle  + \big\langle \g_{\cI_{T-1}'}\odot \bzeta_{\cI_{T}}, \y_{\cI_{T}}\big\rangle +\big\langle \g_{\cI_0}\odot \bzeta_{\cI_1}, \y_{\cI_1}\big\rangle\bigg).
\end{equation}
We now let $\widehat{\g}$ be such random vector that is identical to $\g$ except that for any $i \in \{0\}\cup[T]$
\[
	\widehat{\g_{\cI_i}} = -\g_{\cI_i}, \;\text{if} \; i  \bmod 2 = 0
\]
Let $\widehat{\bzeta}$ be such random vector that is identical to $\bzeta$ except that for any $i \in \{0\}\cup[T]$
\[
	\widehat{\bzeta_{\cI_i}} = -\bzeta_{\cI_i}, \;\text{if} \; i \bmod 2 = 1.
\]
Replacing $\g, \;\bzeta$ in (\ref{Z_formula}) with $\widehat{\g}, \;\widehat{\bzeta}$ yields
\begin{align*}
	\widehat{Z} \\
	= & \bigg( \sum_{i=1}^{T} \big\langle \widehat{\g}_{\cI_i}\odot \widehat{\bzeta}_{\cI_i}, \x_{\cI_i}\big\rangle\bigg) \cdot \bigg( \sum_{i=1}^{T-2}\big\langle \widehat{\g}_{\cI_i}\odot \widehat{\bzeta}_{\cI_{i+1}}, \y_{\cI_{i+1}}\big\rangle  + \big\langle \widehat{\g}_{\cI_{T-1}'}\odot \widehat{\bzeta}_{\cI_{T}}, \y_{\cI_{T}}\big\rangle +\big\langle \widehat{\g}_{\cI_0}\odot \widehat{\bzeta}_{\cI_1}, \y_{\cI_1}\big\rangle\bigg) \\
	=& \bigg( - \sum_{i=1}^{T} \big\langle \g_{\cI_i}\odot \bzeta_{\cI_i}, \x_{\cI_i}\big\rangle\bigg) \cdot \bigg( \sum_{i=1}^{T-2}\big\langle \g_{\cI_i}\odot \bzeta_{\cI_{i+1}}, \y_{\cI_{i+1}}\big\rangle  + \big\langle \g_{\cI_{T-1}'}\odot \bzeta_{\cI_{T}}, \y_{\cI_{T}}\big\rangle +\big\langle \g_{\cI_0}\odot \bzeta_{\cI_1}, \y_{\cI_1}\big\rangle\bigg) \\
	= & -Z.
\end{align*}
As each entry of $\g$ is symmetric random variable around $0$, therefore $\widehat{\g}$ and $\g$ has the same probability distribution. The same fact also holds for $\widehat{\bzeta}$ and $\bzeta$. So we conclude that $Z$ has symmetric distribution around $0$, which implies $\Pr(Z > 0) = \frac{1}{2}$ and $X \bot Y'$.
\end{proof}

\subsubsection{``Proof'' of Theorem \ref{thm:fbe}} \label{proof:thm:fbe}

\begin{proof}
Unspecified notations in this section are consistent with Algorithm \ref{alg:fbe}. Using Lemma \ref{lem:followFastJL}, we have 
\begin{equation} \label{fastJLinequality}
	\Pr\bigg\{\sup_{i,j \in [N]}\big| d(\y_i, \y_j) - d(\x_i,\x_j) \big|  \geq C\delta \bigg\} \leq 0.01.
\end{equation}
Now consider the first-block binary codes generated from Gaussian Toeplitz projection. We focus on two intermediate points $\y_1$ and $\y_2$. Consider the first block of binary codes generated from the second part of Algorithm \ref{alg:fbe}. We let 
\[
	\u = \sign\big(\Psi^{(1)} \cdot \y_1\big), \v = \sign\big(\Psi^{(1)} \cdot \y_2\big).
\]
Suppose $\Psi^{(1)}$ contains Gaussian Toeplitz matrix $\Tb$. For any $i \in [m/B]$, we have
\[
	u_i = \sign\big(\big\langle \Tb_i \odot\bzeta, \;\y_1\big\rangle\big) = \sign\big(\big\langle \Tb_i, \; \y_1\odot\bzeta\big\rangle\big).
\]
\[
	v_i = \sign\big(\big\langle \Tb_i\odot\bzeta, \; \y_2\big\rangle\big) = \sign\big( \big\langle \Tb_i, \;\y_2\odot\bzeta\big\rangle\big).
\]
Since $\Tb_i$ is a Gaussian random vector, we have
\[
	\Pr(u_i \ne v_i) = d(\y_1\odot\bzeta,\; \y_2\odot\bzeta) = d(\y_1,\; \y_2).
\]
Let $Z_i = \mathds{1}\big(u_i \ne v_i\big), \forall \; i \in [m/B]$. Following Lemma (\ref{lem:pairwise_independent}), we know that $\forall \; i \ne j$
\[
	u_i \bot u_j, \; u_i \bot v_j, \; v_i \bot v_j, \; v_i \bot u_j. 
\]
Therefore $\{Z_i\}_{i=1}^{[m/B]}$ is a pair-wise independent sequence. \textbf{[This implication is incorrect, as observed by~\cite{dirksen2018fast}.]}
By Markov's inequality, we have 
\begin{equation} \label{markov}
	\Pr\bigg(\big|\frac{1}{m/B} \sum_{i=1}^{m/B} Z_i - \mathbb{E}(Z_1)\big| \geq \delta \bigg) \leq \frac{\frac{B}{m}Var(Z_1)}{\delta^2} \leq \frac{1}{4}\frac{B}{m\delta^2} \leq \frac{1}{4}.
\end{equation}
The last inequality holds by setting $\frac{m}{B} \geq \frac{1}{\delta^2}$. Therefore, we have
\[
	\Pr\bigg(\big|d_{\cH}(\u,\v) - d(\y_1, \y_2)\big| \geq \delta \bigg) \leq \frac{1}{4}. 
\]

Now consider total $B$ block binary codes $\{\u_i\}_{i=1}^{B}$ $\{\v_i\}_{i=1}^{B}$ from $\y_1$ and $\y_2$ respectively. Let
\[
	E_i = \mathds{1}\big(|d_{\cH}(\u_i, \v_i) - d(\y_1, \y_2)| \geq \delta\big), \;\forall \; i \in [B].
\]
From (\ref{markov}), we have $\Pr(E_i = 1) < \frac{1}{4}$. If more than half of $E_i$ are $0$, then the median of $\{d_{\cH}(\u_i, \v_i)\}_{i=1}^{B}$ is within $\delta$ away from $d(\y_1, \y_2)$. Then we have
\begin{align*}
	& \Pr\bigg(\median\big(\{d_{\cH}(\u_i, \v_i)\}_{i=1}^{B}\big) - d(\y_1, \y_2)\big| \geq \delta\bigg) \\
	& \leq  \Pr\big(\frac{1}{B}\sum_{i=1}^{B}E_i \geq \frac{1}{2}\big) 
	\leq \Pr\big(\frac{1}{B}\sum_{i=1}^{B}E_i - \mathbb{E}(E_i) > \frac{1}{4}\big) \leq \exp(-\frac{1}{4}B). 
\end{align*}
In the second inequality, we use (\ref{markov}). The last step follows from Hoeffding's inequality. Now we use a union bound for $N^2$ pairs 
\[
	\Pr\bigg( \sup_{i,j\in[N]}\big|d_{\cH}(\b_i,\b_j) - d(\y_i,\y_j) \big| \geq \delta\bigg) \leq N^2\exp(-\frac{1}{4}B) \leq \exp(-\frac{1}{8}B).
\]
The last inequality holds by setting $B \geq 16\log N$. Combing the above result and (\ref{fastJLinequality}) using triangle inequality, we complete the proof.
\end{proof}

\subsection{Proof of Theorem \ref{thm:general_K}} \label{proof:thm:general_K}

For any set $K \subseteq \SSS^{p-1}$, we use $\cN_{\delta}(K)$ to denote a constructed $\delta$-net of $K$, which is a $\delta$-covering set with minimum size. In particular, by Sudakov's theorem (e.g., Theorem 3.18 in \citet{ledoux1991probability})
\[
	\log \cN_{\delta}(K) \lesssim \frac{w(K)^2}{\delta^2}.
\] 

We first prove that for a fixed two dimensional space, $m = O(\frac{1}{\delta^2})$ independent Gaussian measurements are sufficient to achieve $\delta$-uniform binary embedding.
\begin{lemma} \label{lem:2dimension_space}
	Suppose $K$ is any fixed two-dimensional subspace in $\SSS^{p-1}$. Let $\Ab\in \RR^{m \times p}$ be a matrix with independent rows $\Ab_i \sim \cN(0, \Ib_p),\;\forall i \in [m]\;$. Suppose $m \geq \frac{1}{\delta^2}\log{\frac{1}{\delta}}$, then with probability at least $1 - 3\exp(-\delta^2m)$,
	\begin{equation} \label{2dimension_space}
		\sup_{\x,\y \in K} \big|d_{\Ab}(\x,\y) - d(\x,\y)\big| \leq C\delta.
	\end{equation}
	Here $C$ is some absolute constant.
\end{lemma}
\begin{proof}
	We postpone the proof to Appendix \ref{proof:2dimension_space}.
\end{proof}

The next lemma shows that the normalized $\ell_1$ norm of $\Ab \x$ provides decent approximation of $\|\x\|_2$.
\begin{lemma} \label{lem:l1}
	Consider any set $K \subseteq \RR^{p}$. Let $\Ab$ be an $m$-by-$p$ matrix with independent rows $\Ab_i \sim \cN(0, \Ib_p)$ for any $i \in [m]\;$. Consider
	\[
		Z = \sup_{\x\in K}\bigg|\frac{1}{m}\sum_{i=1}^m \big|\big\langle \Ab_i,\x\big\rangle\big| - \sqrt{\frac{2}{\pi}}\|\x\|_2\bigg|.
	\]
	We have
	\[
		\Pr\big\{ Z \geq 4\frac{w(K)}{\sqrt{m}} + t\big\} \leq 2\exp\big(-\frac{mt^2}{2d(K)^2}\big), \;\forall \;t > 0.
	\]
where $d(K) = \max_{\x \in K} \|\x\|_2$.
\end{lemma}
\begin{proof} 
	See the proof of Lemma 2.1 in \citet{plan2014dimension}.
\end{proof}

In order to connect $\ell_1$ norm to Hamming distance, we need the following result.
\begin{lemma} \label{lem:absolute}
Consider finite number of points $K \subseteq \SSS^{p-1}$. Let $\Ab$ be an $m$-by-$p$ matrix with independent rows $\Ab_i \sim \cN(0, \Ib_p)$ for any $i \in [m]\;$. Suppose
\[
	m \geq \frac{1}{\delta^2}\log{|K|},
\]
then we have
\[
	\sup_{\x \in |K|} \frac{1}{m}\sum_{i=1}^{m}\mathds{1}\bigg\{\big|\big\langle \Ab_i,\x \big\rangle   \big| \leq \delta\bigg\} \leq 2\delta.
\]
with probability at least $1 - \exp(-\delta^2m)$.
\end{lemma}
\begin{proof}
Let $X \sim \cN(0,1)$. For any fixed point $\x \in K$ and any $i \in [m]$, we have
\[
	\Pr(\big|\big\langle \A_i, \x\big\rangle\big| \leq \delta) = \Pr(|X| \leq \delta ) \leq \delta.
\]
Let $Z_i = \mathds{1}(\big|\big\langle \A_i, \x\big\rangle\big| \leq \delta), \; \forall \; i \in [m]$. Then by using Hoeffding's inequality,
\[
	\Pr(\frac{1}{m}\sum_{i=1}^{m}Z_i - \mathbb{E}(Z_1) > \delta) \leq \exp(-2\delta^2m). 
\]
As $\mathbb{E}(Z_1) = \Pr(\big|\big\langle \A_i, \x\big\rangle\big| \leq \delta) \leq \delta$, we conclude that with probability at least $1 - \exp(-2\delta^2m)$,
\[
	\frac{1}{m}\sum_{i=1}^{m}Z_i \leq 2\delta.
\]
By applying union bound over $|K|$ points and setting $m \geq \frac{1}{\delta^2}\log{|K|}$, we complete the proof.
\end{proof}

Now we are ready to prove Theorem \ref{thm:general_K}.
\begin{proof}[Proof of Theorem~\ref{thm:general_K}]
We construct a $\delta$-net of $K$ that is denoted as $\cN_{\delta}$. We assume $m \gtrsim \frac{1}{\delta^2}\log |\cN_{\delta}|$. Applying Proposition \ref{thm:urp} and setting $K = \cN_{\delta}$, we have that
\begin{equation} \label{basicNet0}
	\sup_{\x,\y \in \cN_{\delta}} \big|d_{\Ab}(\x,\y) - d(\x,\y)\big| \leq \delta
\end{equation}
with probability at least $1 - 2\exp(-\delta^2m)$. 

For any two fixed points $\x,\y \in K$, let $\x_1, \y_1$ be their nearest points in $\cN_{\delta}$. Then we have
\begin{align}
	& |d(\x,\y) - d_{\Ab}(\x,\y)| \leq |d(\x,\y) - d(\x_1,\y_1)| + |d(\x_1,\y_1) - d_{\Ab}(\x,\y)| \notag \\ 
	& \overset{(a)}{\leq}  |d(\x_1,\y_1) - d_{\Ab}(\x,\y)| + 2\delta  
	 \leq  |d_{\Ab}(\x_1,\y_1) - d_{\Ab}(\x,\y)| + |d(\x_1,\y_1) - d_{\Ab}(\x_1,\y_1)| + 2\delta \notag\\
	& \overset{(b)}{\leq} |d_{\Ab}(\x_1,\y_1) - d_{\Ab}(\x,\y)| + 3\delta 
	 \leq |d_{\Ab}(\x_1,\y_1) - d_{\Ab}(\x_1,\y)| + |d_{\Ab}(\x_1,\y) - d_{\Ab}(\x,\y)| + 3\delta \notag\\
	& \overset{(c)}{\leq} d_{\Ab}(\y_1,\y) + d_{\Ab}(\x_1,\x) + 3\delta,
\end{align}
where $(a)$ follows from 
\[
	|d(\x,\y) - d(\x_1,\y_1)| \leq |d(\x,\y) - d(\x_1,\y)| +  |d(\x_1,\y) - d(\x,\y_1)| \leq d(\x,\x_1) + d(\x_1,\y_1) \leq 2\delta,
\]
step $(b)$ follows from \eqref{basicNet0}, step $(c)$ follows from the triangle inequality of Hamming distance. Therefore we have
\begin{equation} \label{eq:tmp1}
	\sup_{\x,\y \in K} \big| d_{\Ab}(\x,\y) - d(\x,\y)\big| \leq 2\sup_{\x_1 \in  \cN_{\delta}}\sup_{\substack{\x \in K \\ \|\x-\x_1\|_2 \leq \delta}}d_{\Ab}(\x,\x_1) + 3\delta.
\end{equation}

Next we bound the tail term 
\[
	T  := \sup_{\x_1 \in  \cN_{\delta}}\sup_{\substack{\x \in K \\ \|\x-\x_1\|_2 \leq \delta}}d_{\Ab}(\x,\x_1).
\]

Recall that 
\[
	K^{+}_{\delta} := K\bigcup \big\{\z \in \SSS^{p-1}: \z = \frac{\x - \y}{\|\x - \y\|_2}, \; \forall \; \x,\y \in K\;\text{if}\;\delta^2 \leq \|\x - \y\|_2 \leq \delta\big\}.
\]
Now we construct a $\delta$-net for $K^+_{\delta} \setminus K$ denoted as $\cN_{\delta}'$. For two distinct points $\x,\y \in \cN_{\delta}'\bigcup \cN_{\delta}$, let $C(\x,\y)$ denote the unit circle spanned by $\x,\y$. We construct $\delta^2$-net $\cC_{\delta^2}(\x,\y)$ for each circle $\cC(\x,\y)$. For simplicity, we just let $\cC_{\delta^2}(\x,\y)$ be the set of points that uniformly split $\cC(\x,\y)$ with interval $\delta^2$. We thus have $|\cC_{\delta^2}(\x,\y)| \lesssim \frac{1}{\delta^2}$. Let $\cG_{\delta}$ denote the union of all circle nets $\cC_{\delta^2}(\x,\y)$ spanned by points in $\cN_{\delta}'\bigcup \cN_{\delta}$, namely
\[
	\cG_{\delta} := \bigcup_{\forall\;\x,\y \in \cN_{\delta}'\bigcup \cN_{\delta}} \cC_{\delta^2}(\x,\y)\cup\{\x,\y\}.
\]
For any point $\x \in K$, we can always find a point in $\cG_{\delta}$ that is $O(\delta^2)$ away from $\x$. To see why the argument is true, we first let $\x_1$ be the nearest point to $\x$ in $\cN_{\delta}$. If $\|\x - \x_1\|_2 \leq \delta^2$, then $\x_1$ is the point we want. Otherwise, we have $\delta^2 \leq \|\x - \x_1\|_2 \leq \delta$. In this case, we have $(\x - \x_1)/\|\x - \x_1\| \in K^{+}$. Following the definition of $K^{+}_{\delta}$, we can always find a point $\x_1' \in \cN_{\delta}'\bigcup \cN_{\delta}$ such that
\begin{equation} \label{nearest_condition}
	\big\| \x_1' - \frac{\x - \x_1}{\|\x - \x_1\|_2}\big\|_2 \leq \delta,
\end{equation}
thereby
\[
	\big\|\x - \underbrace{\big(\|\x - \x_1\|_2\x_1' + \x_1\big)}_{\z}\big\|_2 \leq \delta\|\x - \x_1\|_2 \leq \delta^2.
\]
Note that $\|\z\|_2$ is very close to $1$ because 
\[
	\delta^4 \geq \|\x - \z\|_2^2 \geq \|\z\|_2^2 - 2\big\langle \z,\x\big\rangle + 1 \geq \|\z\|_2^2 - 2\|\z\|_2 + 1 = (\|\z\|_2 - 1)^2.
\] 
We thus have
\[
	\big\|\x - \z/\|\z\|_2\big\|_2 \leq \|\x - \z\|_2 + \big\|\z - \z/\|\z\|_2\big\|_2 =  \|\x - \z\|_2 + \big|\|\z\|_2 - 1\big| \leq 2\delta^2.
\]
Note that $\z$ is in the unit circle $\cC(\x,\x_1')$ spanned by $\x$ and $\x_1'$, thereby there exists $\u \in \cC_{\delta^2}(\x_1,\x_1')$ such that
$\|\u - \x\|_2 \leq \delta^2$. Point $\u$ thus satisfies 
\begin{equation} \label{u we want}
	\|\x - \u\| \leq \|\x - \z\|_2 + \|\z - \u\|_2 \leq 3\delta^2.
\end{equation}
So for any $\x \in K$ and its nearest point $\x_1 \in \cN_{\delta}$, we define $\u$ as
\[
\u := \left\{
  \begin{array}{lr}
     \x_1, & \|\x - \x_1\|_2 \leq \delta^2;  \\
     \argmin_{\v \in \cC_{\delta^2}(\x_1,\x'_1)}\|\x - \v\|_2, & \text{otherwise}.
  \end{array}
  \right.
\]
where $\x_1' \in \cN_\delta \bigcup \cN'_{\delta}$ and satisfies (\ref{nearest_condition}). Based on (\ref{u we want}), we always have $\|\u - \x\|_2 \leq 3\delta^2$ and $\|\u - \x_1\|_2 \leq \|\u - \x\|_2+\|\x-\x_1\|_2 \leq 2\delta$.

By triangle inequality of Hamming distance, 
\[
	d_{\Ab}(\x,\x_1) \leq d_{\Ab}(\x,\u) + d_{\Ab}(\u,\x_1).
\] 
We thus have
\begin{align*}
	T & \leq \sup_{\x_1 \in \cN_{\delta}}\sup_{\substack{\x \in K \\\|\x - \x_1\|_2}}  d_{\Ab}(\x,\u) + d_{\Ab}(\u,\x_1) \\
	&\leq \underbrace{\sup_{\u \in \cG_{\delta}} \sup_{\substack{\x \in K \\ \|\x - \u\|_2 \leq 3\delta^2}} d_{\Ab}(\x,\u)}_{T_1} + \underbrace{\sup_{\x,\y \in \cN_{\delta}\bigcup \cN'_{\delta}}\sup_{\substack{\u,\v \in \cC(\x,\y) \\ \|\u - \v\|_2 \leq 2\delta}} d_{\Ab}(\u,\v)}_{T_2}.
\end{align*}
Next we bound term $T_1$ and $T_2$ respectively.

\noindent {\bf Term $T_1$.}  
For a fixed point $\u \in \cG_{\delta}$, using Lemma \ref{lem:l1} by setting  $(K, t)$ in the statement to be $K' = (K - \{\u\})\bigcap\{\u \in \RR^p: \|\u\|_2 \leq 3\delta^2\}$ and $\delta^2$ respectively yields that 
\begin{align*}
	& \Pr\bigg\{\sup_{\substack{\x \in K \\ \|\x - \u\|_2 \leq 3\delta^2}} \bigg|\frac{1}{m}\sum_{i=1}^{m}\big|\big\langle\Ab_i,\x-\u\big\rangle\big| - \sqrt{\frac{2}{\pi}}\|\x-\u\|_2\bigg| \geq \frac{4w(K')}{\sqrt{m}} + \delta^2\bigg\} \\
	\leq & 2\exp\big(-\frac{m\delta^4}{2d(K')^2}\big) \leq 2\exp(-m/18).
\end{align*}
Then with probability greater than $1 - 2\exp(-m/18)$,
\[
	\sup_{\substack{\x \in K \\ \|\x - \u\|_2 \leq 3\delta^2}} \frac{1}{m}\sum_{i=1}^{m}\big|\big\langle\Ab_i,\x-\u\big\rangle\big| 
	\leq 3\sqrt{\frac{2}{\pi}}\delta^2 + 4w(K')/\sqrt{m} + \delta^2 
	\leq  5\delta^2,
\]
where the last inequality follows from the fact that $w(K') \lesssim w(K)$ and our assumption $m \gtrsim w(K)^2/\delta^4$. We define event
\[
	\cE := \bigg\{\sup_{\u \in \cG_{\delta}} \sup_{\substack{\x \in K \\ \|\x - \u\|_2 \leq 3\delta^2}} \frac{1}{m}\sum_{i=1}^{m}\big|\big\langle\Ab_i,\x-\u\big\rangle\big| \leq 5\delta^2\bigg\}.
\]
Applying union bound over all points in $\cG_{\delta}$, we have
\[
	\Pr( \cE^c ) \leq 2|\cG_{\delta}|\exp(-m/18) \leq 2\exp(-m/36),
\]
where the last inequality holds with $m \gtrsim \log|\cG_\delta|$. Under condition event $\cE$ happens, we have
\begin{equation} \label{temp1}
	\sup_{\u \in \cG_{\delta}}\sup_ {\substack{\x \in K \\ \|\x - \u\|_2 \leq 3\delta^2}} \frac{1}{m} \sum_{i=1}^{m} \mathds{1}\bigg\{\big|\big\langle \Ab_i, \u - \x\big\rangle\big| \leq 5\delta\bigg\} \geq 1 - \delta.
\end{equation}
If $\sign\big(\langle \Ab_i, \u\big\rangle\big) \ne  \sign\big(\langle \Ab_i, \x\big\rangle  \big)$, we must have $\big|\big\langle \Ab_i, \u \big\rangle\big| \leq \big|\big\langle \Ab_i, \u - \x \big\rangle\big|$. We then have
\begin{align*}
	T_1 & \leq \sup_{\u \in \cG_{\delta}}\sup_{\substack{\x \in K \\ \|\x - \u\|_2 \leq 3\delta^2}}\frac{1}{m}\sum_{i=1}^{m}\mathds{1}\bigg\{\big|\big\langle \Ab_i, \u \big\rangle\big| \leq \big|\big\langle \Ab_i, \u - \x \big\rangle\big|\bigg\} \\
	& \leq \sup_{\u \in \cG_{\delta}} \frac{1}{m}\sum_{i=1}^{m}\mathds{1}\bigg\{\big|\big\langle \Ab_i, \u \big\rangle\big| \leq 5\delta\bigg\} + \delta,
\end{align*}
where the last inequality follows from (\ref{temp1}). Using Lemma \ref{lem:l1} by setting $K$ and $\delta$ in the statement to be $\cG_{\delta}$ and $5\delta$ respectively, we have that, when $m \geq c\frac{1}{\delta^2}\log|\cG_{\delta}|$ with some absolute constant $c$, the following inequality 
\[
	\sup_{\u \in \cG_{\delta}} \frac{1}{m}\sum_{i=1}^{m}\mathds{1}\bigg\{\big|\big\langle \Ab_i, \u \big\rangle\big| \leq 5\delta\bigg\} \leq 10\delta
\]
holds with probability at least $1 - \exp(-25\delta^2m)$. Putting all ingredients together, we have $T_1 \leq 11\delta$ with high probability.

\noindent{\bf Term $T_2$.} 
There are at most $|\cN_{\delta}\bigcup \cN'_{\delta}|^2$ different two-dimensional subspaces constructed from $\cN_{\delta}\bigcup \cN'_{\delta}$. Applying Lemma \ref{lem:2dimension_space} and probabilistic union bound over all subspaces yields that
\[
	\Pr\bigg(  T_2 \geq  (C+2)\delta\bigg) \leq 3\big|\cN_{\delta}\bigcup \cN'_{\delta}\big|^2\exp(-\delta^2m) \leq 3\exp(-\delta^2m/2),
\]
where the last inequality holds by setting $m \gtrsim \frac{1}{\delta^2}\log|\cN_{\delta}\bigcup \cN'_{\delta}|$.

Putting \eqref{eq:tmp1} and the upper bounds of term $T$ together,  we conclude that by choosing 
\[
	m \gtrsim \max\bigg\{ w(K)^2/\delta^4,\; \log|\cG_\delta|,\; \frac{1}{\delta^2}\log|\cN_{\delta}\bigcup \cN'_{\delta}|\bigg\},
\]
we have
\[
	\sup_{\x,\y \in K}|d_{\Ab}(\x,\y) - d(\x,\y)| \lesssim \delta.
\]
with probability at least $1 - c_1\exp(-c_2\delta^2m)$ where $c_1,c_2$ are some absolute constants.

Using the fact that 
\[
	|\cG_{\delta}| \lesssim \frac{1}{\delta^2}|\cN_{\delta}\bigcup \cN'_{\delta}|
\]
and 
\[
	\log|\cN_{\delta}\bigcup \cN'_{\delta}| \lesssim \frac{1}{\delta^2}w(\cN_{\delta}\bigcup \cN'_{\delta})^2 \leq \frac{1}{\delta^2}w(K^+_{\delta})^2,
\]
we complete the proof.
\end{proof}

\appendix
\section{Proof of Lemma \ref{lem:followFastJL}} \label{proof:lem:followFastJL}
\begin{proof}
Recall that $\y_i = \sqrt{\frac{p}{m}}\Phi(\bzeta)\cdot\x_i$. We let
\[
	\widehat{\y_i} = \frac{\y_i}{\|\y_i\|_2}, \widehat{\y_j} = \frac{\y_j}{\|\y_j\|_2}. 
\]
From condition (\ref{c2}), we have
\begin{equation}
	\label{tmp23}
	\|\y_i - \widehat{\y_i}\|_2 \leq \delta, \|\y_j - \widehat{\y_j}\|_2 \leq \delta.
\end{equation}
Let $\theta = \angle(\x_i,\x_j)$, $\theta' = \angle(\widehat{\y_i},\widehat{\y_j})$.
Without loss of generality, we assume our set $K = \{\x_i\}_{i=1}^{N}$ is symmetric, i.e., if $\x \in K$ then $-\x \in \K$. Suppose we show for any two points $\x_i,\x_j$ with $\big\langle \x_i,\x_j\big\rangle > 0$, inequality (\ref{bound1}) holds, then for $\x_i, \x_j$ with $\big\langle \x_i,\x_j\big\rangle < 0$, we immediately have
\[
	\big|d(\y_i,\y_j) - d(\x_i,\x_j)\big| = \big|1 - d(\y_i,\y_j) - \big(1 - d(\x_i,\x_j)\big)\big| = \big|d(-\y_i,\y_j) - d(-\x_i,\x_j)\big| \leq C\delta.
\]
In the second equality, we use $d(-\x,\y) + d(\x,\y) = 1,\;\forall\;\x,\y \in \SSS^{p-1}$. In the last inequality, we use the fact that fast JL transform $\sqrt{\frac{p}{m}}\Phi(\bzeta)$ is linear thus $-\y_i =  \sqrt{\frac{p}{m}}\Phi(\bzeta)(-\x_i)$. Therefore, without loss of generality, we assume $\langle\x_i,\x_j\rangle \geq 0$ thus $\theta \leq \frac{\pi}{2}$.
	
Now we turn to the following quantity
\begin{align*}
	\big\|\widehat{\y_i} - \widehat{\y_j}\big\|_2 
	& =   \big\|\widehat{\y_i} - \y_i + \y_i - \y_j + \y_j - \widehat{\y_j}\big\|_2 \\
	& \leq  \big\|\widehat{\y_i} - \y_i\big\|_2 + \big\|\widehat{\y_j} - \y_j\big\|_2 + \big\|\y_i - \y_j\big\|_2 
	\leq  2\delta + \|\x_i - \x_j\|_2(1 + \delta).
\end{align*}
The last inequality follows from (\ref{tmp23}) and condition (\ref{c1}). Similarly, we also have
\[
	\big\|\widehat{\y_i} - \widehat{\y_j}\big\|_2 \geq  \|\x_i - \x_j\|(1 - \delta) - 2\delta.
\]
Using the fact that 
\[
	\sin\frac{\theta'}{2} = \frac{\big\|\widehat{\y_i} - \widehat{\y_j}\big\|_2}{2},
	\sin\frac{\theta}{2} = \frac{\big\|\x_i - \x_j\big\|_2}{2},
\]
we have
\[
	\big| \sin\frac{\theta'}{2} - \sin\frac{\theta}{2}\big| = \big| \frac{\big\|\widehat{\y_i} - \widehat{\y_j}\big\|_2}{2} - \frac{\big\|\x_i - \x_j\big\|_2}{2} \big| \leq \delta + \delta\frac{\big\|\x_i - \x_j\big\|_2}{2} \leq 2\delta.
\]
When $\delta < \frac{\sqrt{3} - \sqrt{2}}{4}$, we have
\[
	\sin\frac{\theta'}{2} \leq \sin\frac{\theta}{2} + \frac{\sqrt{3} - \sqrt{2}}{2} \leq \frac{\sqrt{3}}{2}.
\]
In the last inequality, we use $\sin\frac{\theta}{2} \leq \frac{\sqrt{2}}{2},\;\forall \;\theta \in [0,\pi/2]$. So $\theta'/2 \in [0,\pi/3]$.
Using the fact that, for any two $\theta, \theta' \in [0,\pi/3]$, there exists constant $c$ such that
\[
	\big|\sin\theta - \sin\theta'\big| \geq c\big|\theta - \theta'\big|,
\]
we have that
\[
	\big|\frac{\theta}{2} - \frac{\theta'}{2}\big| \leq \frac{1}{c}\big| \sin\frac{\theta'}{2} - \sin\frac{\theta}{2}\big| \leq \frac{2\delta}{c}.
\]
Therefore,
\[
	\big| d(\y_i, \y_j) - d(\x_i,\x_j) \big| = \frac{1}{\pi}\big|\theta - \theta'\big|\leq C\delta.
\]
In the case $\delta > \frac{\sqrt{3}-\sqrt{2}}{4}$, trivially we have $\big|d(\y_i,\y_j) - d(\x_i,\x_j)\big| \leq 2 \leq C\delta$ with constant $C = \frac{8}{\sqrt{3}-\sqrt{2}}$.
\end{proof}

\section{Proof of Lemma \ref{weak_rank_bound}} \label{proof:lem:weak_rank_bound}
\begin{proof}
For positive semidefinite matrix $\Mb \in \RR^{p \times p}$ with rank $d$, let $\lambda_1,\lambda_2,...\lambda_d$ be its positive eigenvalues. Using the definition of Frobenius norm, we have
\begin{equation*}
	\|\Mb\|_F^2 = \sum_{i=1}^{d} \lambda_i^2  = \sum_{i,j \in [n]} (\Mb_{i,j})^2 \leq p + (p^2- p)\delta_2^2.
\end{equation*}
On the other hand, considering the trace of $\Mb$, we can obtain
\begin{equation} \label{tmpp1}
	\sum_{i=1}^{d} \lambda_i = {\rm Trace}(\Mb) \geq p(1 - \delta_1).
\end{equation}
Using Cauchy-Schwarz inequality, we have
\begin{equation} \label{tmpp2}
	(\sum_{i=1}^{d} \lambda_i)^2 \leq d \sum_{i=1}^{d} \lambda_i^2.
\end{equation}
Plugging (\ref{tmpp1}) and (\ref{tmpp2}) into the above inequality yields
\[
	d \geq \frac{p(1 - \delta_1)^2}{1 + (p- 1)\delta_2^2}.
\]
\end{proof}

\section{Proof of Lemma \ref{2dimension_space}} \label{proof:2dimension_space}
\begin{proof}
Without loss of any generality, we assume $K = \{\x \in \SSS^{p-1}: \supp(\x) \subseteq \{1,2\} \}$. We begin with constructing a $\delta$-net denoted as $\cN_{\delta}$ for set $K$. For simplicity, we can just let $\cN_{\delta}(K)$ be the set of points that split the circle spanned by $\{\e_1,\e_2\}$ uniformly. Therefore $|\cN_{\delta}(K)| = O(\frac{1}{\delta})$. Applying Proposition \ref{thm:urp} gives us 
\begin{equation} \label{basicNet}
	\sup_{\x,\y \in \cN_{\delta}} |d_A(\x,\y) - d(\x,\y)| \leq \delta,
\end{equation}
holds with probability at least $1 - 2\exp{(-\delta^2m)}$ when $m \gtrsim \frac{1}{\delta^2}\log(\frac{1}{\delta})$.

For any point $\x \in K$, $\big\langle\Ab_i,\x\big\rangle$ only depends on the first two coordinates of $\Ab_i$. Therefore, for simplicity, we let $\Ab'_{i} = \frac{\Ab_i\odot(\e_1 + \e_2)}{\|\Ab_i\odot(\e_1 + \e_2)\|_2},\; \forall\; i \in [m]$. For any point say $\x_1 \in \cN_{\delta}$, using the uniform distribution of $\Ab'_i$, we have
\[
	\Pr( \big|\big\langle \Ab_i', \x_1\big\rangle\big| \leq \delta) \lesssim C\delta,
\]
holds with some absolute constant $C$. Using Hoeffding's inequality and probabilistic union bound over all points in $\cN_{\delta}$, we have
\begin{equation} \label{distanceBound}
	\Pr\bigg( \sup_{\x \in \cN_{\delta}} \frac{1}{m}\sum_{i=1}^{m}\mathds{1}\big\{\big|\big\langle \Ab_i, \x\big\rangle\big| \leq \delta\big\} > (C+1)\delta\bigg) \leq |\cN_{\delta}|\exp(-2\delta^2m) \leq \exp(-\delta^2m).
\end{equation}
The last inequality holds when $m \gtrsim \frac{1}{\delta^2}\log\frac{1}{\delta}$. 

Now we consider any point $\x \in K$. Suppose $\x_1$ is the closest point to $\x$ in $\cN_{\delta}$. We note that if $\sign\big(\big\langle \Ab_i', \x\big\rangle\big) \ne \sign\big(\big\langle \Ab_i', \x_1\big\rangle\big)$, then there exists $\lambda \in [0,1]$ such that
\[
	\big\langle \Ab_i', \lambda\x + (1 - \lambda)\x_1\big\rangle = 0.
\]
We thus have 
\[
	\big|\big\langle\Ab_i',\x_1 \big\rangle\big| = \lambda \big|\big\langle\Ab_i',\x - \x_1 \big\rangle\big| \leq \lambda \|\x-\x_1\|_2 \leq \delta.
\]
Further we obtain that
\begin{align*}
	 \sup_{\x_1 \in \cN_{\delta}}\sup_{\substack{\x \in K \\ \|\x-\x_1\|_2\leq \delta} } d_{\Ab}(\x,\x_1) 
	& =  \sup_{\x_1 \in \cN_{\delta}}\sup_{\substack{\x \in K \\ \|\x-\x_1\|_2\leq \delta} }   \frac{1}{m}\sum_{i=1}^{m}\mathds{1}(\sign\big(\big\langle \Ab_i', \x\big\rangle\big)  \ne  \sign\big(\big\langle \Ab_i', \x_1\big\rangle\big)) \\
	& \leq \sup_{\x_1 \in \cN_{\delta}}\frac{1}{m}\sum_{i=1}^{m}\mathds{1}\big\{\big|\big\langle \Ab_i, \x_1\big\rangle\big| \leq \delta\big\}.
\end{align*}
Combining the above result with (\ref{distanceBound}), we obtain that, with probability at least $1 - \exp(-\delta^2m)$,
\begin{equation} \label{smalltailbound}
	\sup_{\x_1 \in \cN_{\delta}}\sup_{\substack{\x \in K \\ \|\x-\x_1\|_2\leq \delta} } d_{\Ab}(\x,\x_1) \leq (C+1)\delta.
\end{equation}

For any points $\x,\y \in K$, let $\x_1,\y_1$ be their nearest points in $\cN_{\delta}$. We have
\begin{align*}
	& |d(\x,\y) - d_{\Ab}(\x,\y)|  \leq |d(\x,\y) - d(\x_1,\y_1)| + |d(\x_1,\y_1) - d_{\Ab}(\x,\y)| \notag \\ 
	& \overset{(a)}{\leq}  |d(\x_1,\y_1) - d_{\Ab}(\x,\y)| + 2\delta  \notag \leq |d(\x_1,\y_1) - d_{\Ab}(\x_1,\y_1)| + |d_{\Ab}(\x_1,\y_1) - d_{\Ab}(\x,\y)| + 2\delta \notag\\
	& \overset{(b)}{\leq} |d_{\Ab}(\x_1,\y_1) - d_{\Ab}(\x,\y)| + 3\delta  \leq |d_{\Ab}(\x_1,\y_1) - d_{\Ab}(\x_1,\y)| + |d_{\Ab}(\x_1,\y) - d_{\Ab}(\x,\y)| + 3\delta \notag\\
	& \overset{(c)}{\leq} d_{\Ab}(\y_1,\y) + d_{\Ab}(\x_1,\x) + 3\delta \overset{(d)}{\leq} (2C+5)\delta,
\end{align*}
where $(a)$ follows from 
\[
	|d(\x,\y) - d(\x_1,\y_1)| \leq |d(\x,\y) - d(\x_1,\y)| +  |d(\x_1,\y) - d(\x,\y_1)| \leq d(\x,\x_1) + d(\x_1,\y_1) \leq 2\delta,
\]
step $(b)$ follows from \eqref{basicNet}, step $(c)$ follows from the triangle inequality of Hamming distance, step $(d)$ is from \eqref{smalltailbound}. 
\end{proof}

{\small
\bibliographystyle{plainnat}
\bibliography{binary_embedding}
}

\end{document}